\documentclass[journal,onecolumn]{IEEEtran}
\usepackage{amsmath}
\usepackage{amssymb, amsfonts}
\usepackage{amsthm}
\usepackage{braket}
\usepackage[colorlinks=true, citecolor=red, urlcolor=blue ]{hyperref}
\usepackage{color,empheq}
\usepackage{xcolor}
\usepackage{float}
\usepackage{epstopdf}
\usepackage{ulem}
\DeclareMathOperator{\tr}{tr}
\usepackage{enumitem}

\let\svitem\item%
\def\mybox#1{\makebox[2.5cm][l]{\bfseries#1}}

\usepackage{array}
\newcolumntype{L}[1]{>{\raggedright\let\newline\\\arraybackslash\hspace{0pt}}m{#1}}
\newcolumntype{C}[1]{>{\centering\let\newline\\\arraybackslash\hspace{0pt}}m{#1}}
\newcolumntype{R}[1]{>{\raggedleft\let\newline\\\arraybackslash\hspace{0pt}}m{#1}}

\newtheorem{remark}{Remark}

\newtheorem{lemma}{Lemma}
\newtheorem{corollary}{Corollary}

\newtheorem{theorem}{Theorem}
\newtheorem{definition}{Definition}
\newtheorem{observation}{Observation}

\newtheorem{notation}{Notation}

\def\bed{\begin{definition}}
	\def\eed{\end{definition}}
\def\bel{\begin{lemma}}
	\def\eel{\end{lemma}}
\def\bet{\begin{theorem}}
	\def\eet{\end{theorem}}
\def\bet{\begin{notation}}
	\def\eet{\end{notation}}
\def\bet{\begin{remark}}
	\def\eet{\end{remark}}

\usepackage{cite}

\ifCLASSINFOpdf
  \usepackage[pdftex]{graphicx}
\else
\fi
\begin{document}
%
\title{Secure communication over generalised quantum multiple access channels}
%
%
%

\author{Tamoghna Das, 
        Karol Horodecki,
        and~Robert Pisarczyk
\thanks{Tamoghna Das was with National Quantum Information Centre in Gda\'nsk, Faculty of Mathematics, Physics and Informatics, now associated with International Centre for Theory of Quantum Technologies, University of Gdańsk, Wita Stwosza 63, 80-308 Gdańsk, Poland, email: tamoghna.das@ug.edu.pl.}
\thanks{Karol Horodecki is with Institute of Informatics
		and National Quantum Information Centre in Gda\'nsk, Faculty of Mathematics, Physics and Informatics, University of Gda\'nsk, 80--952 Gda\'nsk, Poland and
		International Centre for Theory of Quantum Technologies, University of Gdańsk, Wita Stwosza 63, 80-308 Gdańsk, Poland, email: khorodec@inf.ug.edu.pl.}
\thanks{Robert Pisarczyk is with the Mathematical Institute, University of Oxford, Woodstock Road, Oxford OX2 6GG, United Kingdom, email: robert.pisarczyk@maths.ox.ac.uk}}

\maketitle

\begin{abstract}
We investigate the security of generalized quantum multiple-access channels. We provide the formula for the achievable rate region of secure communication in the scenario of two senders and a single receiver. We explicitly specify a protocol for secure communication in this scenario, which employs superdense coding. The protocol is based on the distribution of a tripartite GHZ state. It allows for both symmetric and asymmetric key distribution whereby one of the senders can have twice the capacity of the other sender. We prove the security of the protocol against general quantum attacks, 
analyze different strategies of the eavesdropper and compute the key rate for a range of noisy quantum channels.
\end{abstract}


%
\IEEEpeerreviewmaketitle

\section{Introduction}
\IEEEPARstart{T}{he} 
study of communication protocols is crucially important from the perspective of building quantum internet and the development of quantum technologies \cite{Muralidharan2016,Pompili_2021, kimble2008quantum} (see \cite{Wehner2018} and references therein). Multipartite settings are in particular yet to be fully understood due to intrinsically quantum effects such as entanglement \cite{Azuma2021,Carrara2021}.

Arguably, one of the most basic multipartite setups is that multiple access channels (MAC) with two senders and a single receiver. 
MACs and their capacity regions have been studied in the early years of quantum information theory \cite{winter2001capacity, hsieh2008entanglement,yard2008capacity}. Recently, it has been shown that finding the capacity region even in that simplest scenario is NP-hard and that entanglement increases its classical capacity \cite{leditzky2020playing}. Here, we go beyond these results and consider the security aspects of MACs.

We consider a generalized MAC scenario with two senders ($B_1$ and $B_2$) and a single receiver  $A=A_1A_2$. In this scheme, the state is first distributed among all the three parties. Next, the two senders perform
quantum operations on their shares and resend
them to the receiver. The quantum MAC part of the channel is followed by classical feedback from the receiver to the senders. We refer to this scenario as a {\it generalized multiple access channel scenario} (GMAC), which is depicted in figure \ref{fig:setup}.

\begin{figure}[t]
\centering
 \includegraphics[width=1\columnwidth,keepaspectratio,angle=0]{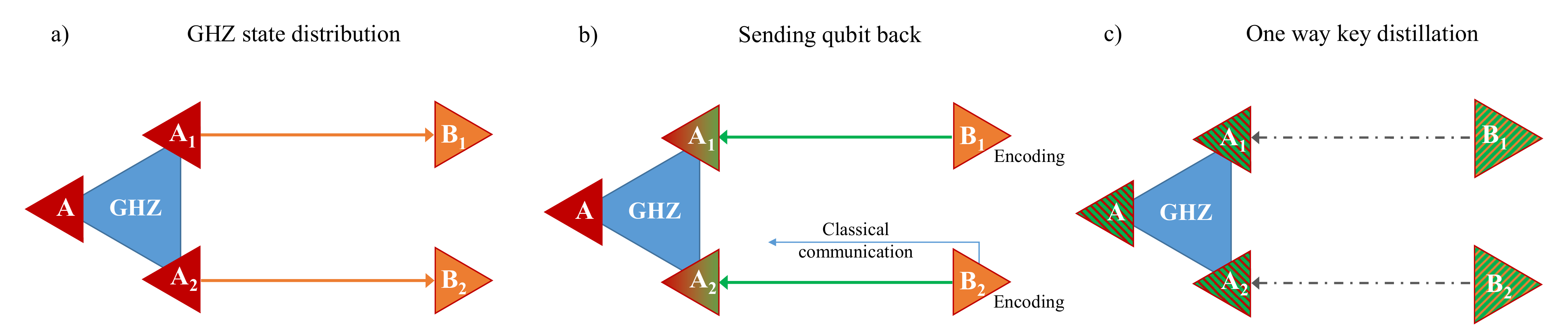}
\caption{ The presented GMAC protocol has three phases: a) Alice distributes two subsystems of a three-qubit GHZ state among senders (Bobs1 and Bob2) and keeps one qubit with herself. b) Both the senders perform encoding with unitary transformations, send back the qubits to Alice, after which Bob2 announces $1$ bit of his encoding. c) Alice performs key distillation protocols with  Bob1 and separately with Bob2, in both cases using one-way classical communication feedback. 
}
\label{fig:setup}
 \end{figure}

In our first main result, we show that in the the above scenario, given that the parties observe (i) high correlations $I(A_1:B_1)\geq H(A_1) -\delta_1$ and $I(A_2:B_2)\geq H(A_2) -\delta_2$ and (ii) low correlations between
systems $B_1$ and $B_2$ ($I(B_1:B_2)\leq \delta_3$), there is a protocol for them to achieve the following key rates between $A_1:B_1$ and $A_2:B_2$ respectively:
\begin{align}
    r({\cal P}_1) \geq  I(A_1:B_1)- I(A_1:E|B_2)- \delta_1 - 2\delta_2-\delta_3\\
    r({\cal P}_2) \geq  I(A_2:B_2)- I(A_2:E|B_1) -2\delta_1-\delta_2-\delta_3
\end{align}
In the above, $I(.)$ stands for quantum mutual information, and $H(.)$ for the Shannon entropy.
The bounds realizing the above rates
are based on the fundamental one-shot lower-bounds
by Joseph Renes and Renato Renner
\cite{RR}.
Given the general formulae, it remains to provide the honest parties a way to lower bound the RHS of the above quantities based on their statistics. Hence, as our second main result, we present such a protocol and study its rate for selected channels.

A trivial example of a GMAC channel is 
that of two separate channels, each is realising a separate protocol with the receiver. However, in some cases, the receiver node may have low computational capabilities. Then one should aim to reduce the number of qubits processed in that node. To effect this, we base our protocol on superdense coding \cite{BennettWisener}. 

Superdense coding is one of the purely quantum effects that use the phenomenon of entanglement and the unique geometry of quantum states \cite{BennettWisener}. It allows doubling of classical capacity between a sender and a receiver is given access to quantum memory used to store entangled states before the communication rounds. Furthermore, quantum dense coding protocol can be made secure when modified by suitable procedures \cite{Beaudry2013}. This has been shown in the case of a single sender single receiver scenario against a quantum adversary.

Such a protocol requires quantum memory
at Alice's side. A non-trivial example of a GMAC, which we further consider, is when the receiver has this memory restricted.
We therefore refer to the multipartite dense-coding effect \cite{BrussALMSS2005-multidense}. When sharing a multipartite GHZ state \cite{GHZk}, one out of $n$ senders can communicate to the receiver at a
double rate ($2$ bits per run), while the others have a connection at the rate of $1$ bit per run. Such multipartite scenarios have recently attracted significant interest due to the application in quantum internet, and the study of secure communication over quantum networks \cite{kimble2008quantum}. Here, we present a secure multipartite protocol for quantum key distribution that allows a sender Alice to share secret keys with two receivers Bob1 and Bob2. In doing so, we generalize the security proof of (suitably modified) secure dense coding protocol (SDC) to two senders and a single receiver. 

Beyond quantum internet, the protocol can be applied to mobile quantum networks such as quantum communication between ground stations and satellites \cite{yin2020entanglement}, airplanes \cite{pugh2017airborne, nauerth2013air} and drones \cite{hill2017drone, liu2020drone}. In particular, one can picture our scenario as a single station and two drones looking for some object on the Earth. They are sent to two different areas for searching. The first to find the object should have a better connection to send the data faster. Our protocol is symmetric, meaning that both drones can, upon agreement, have the privilege to obtain higher channel capacity (see figure \ref{fig:drones}).

Security proof of the dense coding protocol in the single sender and single receiver case, as shown in \cite{Beaudry2013}, is based on two important facts. The first one is the uncertainty relation for complementary observables \cite{Berta2010} that are used in the protocol (generalizations of $X$ and $Z$
measurements of the famous BB84 
protocol \cite{BB84}). The second one leverages the observation that every measurement-based protocol can
be purified - the technique introduced by Shor and Preskill \cite{shor-preskill}. We exploit these ideas; however, the proposed protocol has an intrinsically multipartite part. Namely, one of the two senders (Bobs) sends a single bit of information to allow the receiver (Alice) the unambiguous decoding. The protocol can be made symmetric with respect to the senders. Namely, by the time-sharing technique, the two senders can
reach (in the ideal case) any two capacities whose sum is $3$ bits, and maximal of which is less than or equal to $2$.
The protocol leverages GHZ states of this protocol
is that there is only $1$ ``e-bit''
of entanglement between receiver and senders in place of $2$ e-bits when two separate links are set. It saves not only entanglement but also local quantum memory.


\subsection{Outline:}
The paper is organized as follows. Section \ref{sec:notation} is devoted to notation and known facts. In Section \ref{sec:DC_def}, we define the dense coding protocol in the multipartite scenario  of two senders and a single receiver. In Section
\ref{sec:secure-dc} we present the $2$-$1$ secure key distribution protocol based on dense coding.
Section
\ref{sec:purification_protocol} is devoted to the presentation of the purified version of the $2$-$1$ this secure dense coding protocol.
The closed formula for the key rate of the protocol is derived further in Section \ref{sec:derivation_of_key}.
We then consider special attacks (noise)
on quantum channels between senders and
the receiver in Section \ref{sec:noisy_case}.  We conclude
with the Discussion Section \ref{sec:discussion}.

\begin{figure}[h]
\centering
 \includegraphics[width=0.7\columnwidth,keepaspectratio,angle=0]{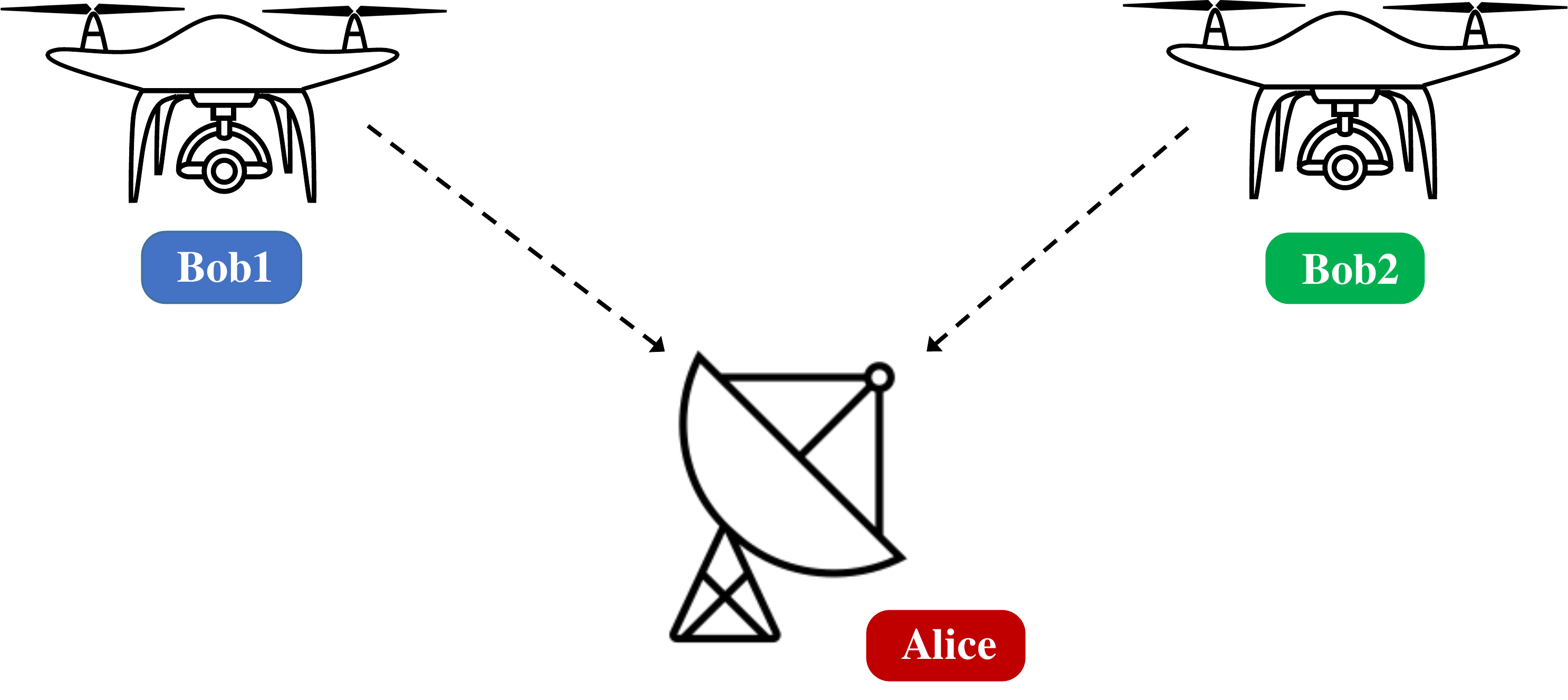}
\caption{ A possible variant
of the secure dense coding protocol with two senders (here depicted as drones) and a single receiver on the ground.}
\label{fig:drones}
 \end{figure}

\section{Preliminaries and notation}\label{sec:notation}
\subsection{Pauli Matrices}
In the  Hilbert space ${\cal H}^2$, of dimension $2$, the three Pauli matrices $\sigma_i$, $i = x,y,z$ are defined as
\begin{equation}\label{eq:Pauli}
 \sigma_x=\left(
\begin{array}{cc}
0 & 1 \\
1 & 0
\end{array}\right)
,~~~~~~~~~~
\sigma_y=\left(
\begin{array}{cc}
0 & -i \\
i & 0
\end{array}\right)
,~~~~~~~~~~
\sigma_z=\left(
\begin{array}{cc}
1 & 0 \\
0 & -1
\end{array}\right).
\end{equation}
These are self adjoint operator and satisfy $\{\sigma_i,\sigma_j\} = 2\delta_{ij}$.
\subsection{Identity operator}
The identity operator $\mathbb{I}_2$, in  ${\cal H}^2$, is defined as $\mathbb{I}_2 =\left(
\begin{array}{cc}
1 & 0 \\
0 & 1
\end{array}\right)$.
\subsection{Mutually orthogonal unitary operators} 
A set of unitary operators $\{W_i\}$, with $W_i^\dagger W_i = \mathbb{I}, ~\forall i$, is called mutually orthogonal \cite{Hiroshima_2001}, if it satisfies
\begin{eqnarray}
\text{tr}(W_iW_j^\dagger) &=& d ~\delta_{ij}, \\
\sum_i W_i \Xi W_i^\dagger &=& d ~\text{tr} (\Xi),
\end{eqnarray} 
for any operator $\Xi$. Here $d$ is the dimension of the Hilbert space they belong to. 
Pauli matrices together with $\mathbb{I}_2$ form a set of mutually orthogonal unitary operators in ${\cal H}^2$.
\subsection{Shannon Entropies and Mutual information \cite{Shannon48,Nielsen-Chuang}}
Suppose $A$ and $B$ are two random variable with  a joint probability distribution $p(a,b)$, where $a \in A$ and $b \in B$, satisfying $\sum_{a,b} p(a,b) = 1$. The  Shannon entropies are defined as 
\begin{equation}
  H(A,B) = -\sum_{a,b} p(a,b)\log_2 p(a,b),
\end{equation} 
and 
\begin{equation}
    H(A) = -\sum_{a} p(a)\log_2 p(a),
\end{equation}
where $p(a) = \sum_b p(a,b)$, similarly for $H(B)$. The Shannon entropy quantifies the amount of ``uncertainty" inherent to the given random variable. \\
Mutual information of two random variable $A$ and $B$ is a measure of the mutual dependence between the two variables and is defined as
\begin{eqnarray}
   I(A:B) &=& H(A) + H(B) - H(A,B)  \\
   &=& H(A) - H(A|B) = H(B) - H(B|A),
\end{eqnarray}
 where $H(A|B)$, is the conditional Shannon  entropy of the random variables $A$ and $B$ and quantifies the randomness remaining in $A$ upon knowing the value of $B$ and is given by 
\begin{eqnarray}
   H(A|B) =  H(A,B)  - H(B).
\end{eqnarray}
\subsection{von Neumann Entropies and Quantum Mutual information \cite{Nielsen-Chuang}}
For a quantum state $\rho_A \in {\cal H}^A$, the von Neumann entropy $S(A)_{\rho_A}$, is defined as 
\begin{equation}
    S(A)_{\rho_A} = - \tr(\rho_A \log_2 \rho_A).
\end{equation}
For a bipartite quantum state $\rho_{AB} \in {\cal H}^A \otimes {\cal H}^B$, the quantum mutual information, $I(A:B)_{\rho_{AB}}$ is 
\begin{eqnarray}
   I(A:B)_{\rho_{AB}} &=& S(A)_{\rho_A} +S(B)_{\rho_B} -S(AB)_{\rho_{AB}} \\
   &=& S(A)_{\rho_A} - S(A|B)_{\rho_{AB}} = S(B)_{\rho_B} - S(B|A)_{\rho_{AB}},
\end{eqnarray}
where $S(AB)_{\rho_{AB}}$ is the von Neumann entropy of the bipartite state $\rho_{AB}$ and $S(A|B)_{\rho_{AB}}$ denotes  conditional von Neumann entropy, defined as 
\begin{equation}
    S(A|B)_{\rho_{AB}} =  S(AB)_{\rho_{AB}} - S(B)_{\rho_B}.
\end{equation}
\subsection{Entropic uncertainty relation}
Following \cite{Beaudry2013} we utilise a fundamental result called the {\it entropic uncertainty relation} \cite{Berta2010, tomamichel2012framework}. Namely for 
any two POVMs $\{M_z\}_z$ and
$\{M_x\}_x$ of a measurement applied to system $A$ on any state $\rho_{ABE}$ and resulting in
classical-quantum states 
$\rho_{XE}=\tr_B (M_x\otimes \mathrm{I}_{BE}\rho_{ABE})$ and $\rho_{ZB}=\tr_E( M_z\otimes\mathrm{I}_{BE}\rho_{ABE})$ 
we have:
\begin{equation}
S(Z|B)_{\rho_{ZB}} + S(X|E)_{\rho_{XE}} \geq \log_2 \frac{1}{\gamma}    
\end{equation}
where $\gamma= \max_{x,z}||M_z M_x||_{\infty}^2$. Here $||X||_\infty=\max \{\bra{\psi}X\ket{\psi}: \bra{\psi}\ket{\psi}=1\}$. 

\subsection{Classical-classial Quantum state}
In the analysis of the security of the protocol, the states under consideration are
partially classical, i.e., diagonal in the computational basis
in two subsystems and quantum in the other. We refer to such states as  ``classical-classical-quantum'' states or shortly ccq states. 
$\sum_{ijk} p_{ijk} \ket{ijk}\bra{ijk}_{AB_1B_2}\otimes \rho^{(ijk)}_E$ is an example of a ccq state with the two
classical systems held by $A$ and collectively $B_1B_2$ while Eve possesses a quantum system.

\subsection{Notations}
\begin{itemize}
\item Throughout the manuscript, we use a basis in the Hilbert space of 3-qubit system, where GHZ state is one of the elements. However, we use two different labelings of its elements namely, $\{\ket{G^s(i,j,k)}\}_{i,j,k = 0}^1$ for different values of $s = 0,1$, where 
\begin{equation}
    \ket{G^s(i,j,k)} = \frac{1}{\sqrt{2}} \sum_{l = 0}^1 (-1)^{l.(j \oplus s)}\ket{l,j\oplus l,k \oplus l},
\end{equation}
\item Four orthogonal Bell states have been considered as the basis for the 2-qubit system, with the following compact notation
\begin{equation}
    \ket{B(x,y)} = \frac{1}{\sqrt{2}} \sum_{l = 0}^1 (-1)^{l.y} \ket{l, x\oplus l}.
\end{equation}

\item Due to the nature of the dense coding protocol, the qubits between senders and the receiver travel twice - forth and back.
We model this as
quantum channels acting on
the transmitted qubits as follows:
\begin{itemize}
    \item {\it Forward transmission channel}: Quantum channel ${\cal E}^f_{\substack{A_1 \rightarrow B_1\\ A_2 \rightarrow B_2}} $, from Alice to the senders used to transmit the two subsystems of a  3-qubit GHZ state.

    \item {\it Backward transmission channel}: Quantum channel ${\cal E}^b_{\substack{B_1 \rightarrow A_1\\ B_2 \rightarrow A_2}} $, from  the senders back to Alice used to transmit the encoded 3-qubit noisy GHZ state or the measured state.
\end{itemize}
\item In this manuscript the eigenvectors of $\sigma_x$, $\{\ket{\pm} = \frac{1}{\sqrt{2}}(\ket{0} \pm \ket{1}) \}$ are denoted in a compact notation $\{\ket{\alpha_\vdash} = \frac{1}{\sqrt{2}}(\ket{0}+ (-1)^\alpha \ket{1})  \}$ for $\alpha \in \{0,1\}$, and the eigenvectors of $\sigma_z$ in the standard form of computational basis $\{\ket{\alpha}, \alpha \in \{0,1\}\}$.
\end{itemize}

\section{Multiparty Dense coding protocol: Two senders and a single receiver}\label{sec:DC_def}

  In this section, we shall discuss the multiparty quantum dense coding (DC) \cite{BennettWisener} protocol involving two senders and a single receiver in greater detail. Both the senders, called  Bob1 and Bob2,  are in a space-like separated location, and hence their encoding procedures are local. Suppose  the senders intend to send some classical information to a common receiver, Alice. The amount of information one sender wants to send is double the length of the other sender. We assume Bob1 wants to communicate a  message of length 2-bit, without loss of generality, whereas Bob2 needs to send only one bit of information.
   In such a situation, a three-party quantum superdense coding protocol proceeds as follows: \cite{BrussALMSS2005-multidense,ShadmanNoise}, 
   \begin{itemize}
       \item A three-qubit GHZ state, \cite{GHZk}, $\ket{GHZ}_{AB_1B_2} = \frac{1}{\sqrt{2}}(\ket{000} + \ket{111})_{AB_1B_2}$ is shared between Alice, Bob1 and Bob2\footnote{The three-qubit GHZ state is the optimal quantum state of the lowest dimension required to be shared between two distant sender wishing to send three bits of classical information to a common receiver using a DC protocol. This follows from the fact that the DC capacity of GHZ state has been proven to be  $3$ \cite{BennettWisener, Hiroshima_2001}}.
       \item  Bob1 encodes his two bits $(x,y)$, where $x,y \in \{0,1\}$, in his part of shared GHZ state, by performing unitary operation out of the set of four operations labeled $U^{xy}$.
        \item  Bob2, the other sender, wishes to send a single bit $z \in \{0,1\}$ and chooses a unitary $U^z$,  to encode his bit $z$, on his part of the shared  GHZ. 
        \item After performing the unitary encoding, both the senders send their encoded part of the shared GHZ, to Alice via a quantum channel. 
        \item Alice receives all the subsystems of the shared state, which is now a 3-qubit encoded state $\ket{G(x,y,z)}$,  given by
        \begin{equation}
            \ket{G(x,y,z)} = \mathbb{I}_2 \otimes U^{x,y} \otimes U^{z} \ket{GHZ}.
        \end{equation}
        She performs a  joint measurement on the total system, 
         and decodes the messages $(x,y,z), ~x,y,z \in \{0,1\}$, from her measurement outcome.  
   \end{itemize}
   
   \begin{table}
    \caption {This table shows the decoding procedure by Alice, when she performs a measurement in the complete ${\cal B}_{GHZ}$ basis, and Bob2 discloses his auxiliary  bit  $s = 0$. The first column shows Alice's measurement outcome, the second column and third column shows Bob1's and Bob2's information that they aim to communicate.}
\label{Table:Secure_DC}
\centering
\renewcommand{\arraystretch}{2.1}
  \begin{tabular}{|>{\centering}m{5cm}|  >{\centering}m{2.2cm}|  >{\centering\arraybackslash}m{2.2cm}|} 
  \hline
  \bf Measurement Outcome & \multicolumn{2}{c|}{\bf Decoded in part of } \\
  \hline
  Alice & Bob1 & Bob2 \\
  \hline
  \hline
   $\frac{1}{\sqrt{2}}(|000\rangle + |111\rangle)_{AB_1B_2}$ & 0~0  & 0   \\ 
   \hline
   $\frac{1}{\sqrt{2}}(|000\rangle - |111\rangle)_{AB_1B_2}$ & 0~1  & 0   \\ 
   \hline
   $\frac{1}{\sqrt{2}}(|010\rangle + |101\rangle)_{AB_1B_2}$ & 1~0  & 0   \\ 
   \hline
   $\frac{1}{\sqrt{2}}(|010\rangle - |101\rangle)_{AB_1B_2}$ & 1~1  & 0   \\ 
 \hline
 \hline
   $\frac{1}{\sqrt{2}}(|001\rangle + |110\rangle)_{AB_1B_2}$ & 0~0  & 1   \\ 
   \hline
    $\frac{1}{\sqrt{2}}(|001\rangle - |110\rangle)_{AB_1B_2}$ & 0~1  & 1   \\ 
   \hline
    $\frac{1}{\sqrt{2}}(|011\rangle + |100\rangle)_{AB_1B_2}$ & 1~0  & 1   \\ 
   \hline
    $\frac{1}{\sqrt{2}}(|011\rangle - |100\rangle)_{AB_1B_2}$ & 1~1  & 1   \\ 
   \hline
  \end{tabular}
  \end{table}

 Note that the messages can be sent without any ambiguity only if both the senders encode their messages by a suitable choice of the unitary operators in such a way that the total shared states for different messages become mutually orthogonal \cite{BennettWisener}. If the transmission channels of the encoded states are noisy, then the two senders (encoder) and the receiver (decoder) need to optimise over all possible encoding and decoding procedures. \cite{Hiroshima_2001,HoroCapacity,BrussALMSS2005-multidense,ShadmanNoise, DDCReznik,ROYDDC}.
   
   Let us start with the simplified situation, in which there is no noise in the system. Hence, Alice, Bob1, and Bob2 share a pure GHZ state. Moreover, the encoded state is also transferred back through a noiseless quantum channel. Then, one of the possible choices of unitary encodings is any set of mutually orthogonal unitary operators in the Hilbert space of qubit ${\cal H}^2$ \cite{Hiroshima_2001}. In particular, one can choose the three Pauli matrices along with the identity operator $\mathbb{I}_2$. 

  \begin{table}
   \caption{This table shows the decoding procedure by Alice, when she performs measurement in the complete ${\cal B}_{GHZ}$ basis, and Bob2 discloses his auxiliary  bit  $s = 1$. The first column shows Alice's measurement outcome, the second column and third column shows Bob1's and Bob2's information they want to communicate. The entire table matches with Table \ref{Table:Secure_DC}, if Bob1's second bit is replaced by $y \rightarrow y \oplus 1$.}
\label{Table:Secure_DC2}
\centering
\renewcommand{\arraystretch}{2.1}
  \begin{tabular}{|>{\centering}m{5cm}|  >{\centering}m{2.2cm}|  >{\centering\arraybackslash}m{2.2cm}|} 
  \hline
  \bf Measurement Outcome & \multicolumn{2}{c|}{\bf Decoded in part of } \\
  \hline
  Alice & Bob1 & Bob2 \\
  \hline
  \hline
   $\frac{1}{\sqrt{2}}(|000\rangle + |111\rangle)_{AB_1B_2}$ & 0~1  & 0   \\ 
   \hline
   $\frac{1}{\sqrt{2}}(|000\rangle - |111\rangle)_{AB_1B_2}$ & 0~0  & 0   \\ 
   \hline
   $\frac{1}{\sqrt{2}}(|010\rangle + |101\rangle)_{AB_1B_2}$ & 1~1  & 0   \\ 
   \hline
   $\frac{1}{\sqrt{2}}(|010\rangle - |101\rangle)_{AB_1B_2}$ & 1~0  & 0   \\ 
 \hline
 \hline
   $\frac{1}{\sqrt{2}}(|001\rangle + |110\rangle)_{AB_1B_2}$ & 0~1  & 1   \\ 
   \hline
    $\frac{1}{\sqrt{2}}(|001\rangle - |110\rangle)_{AB_1B_2}$ & 0~0  & 1   \\ 
   \hline
    $\frac{1}{\sqrt{2}}(|011\rangle + |100\rangle)_{AB_1B_2}$ & 1~1  & 1   \\ 
   \hline
    $\frac{1}{\sqrt{2}}(|011\rangle - |100\rangle)_{AB_1B_2}$ & 1~0  & 1   \\ 
   \hline
  \end{tabular}
  \end{table}

Suppose Bob1 chooses the following unitary operators 
\begin{equation}\label{Eq:encode_uni}
    \left.\begin{array}{l}
        U^{00} = \mathbb{I}_2 \\
        U^{01} = \sigma_z\\
        U^{10} = \sigma_x \\
        U^{11} = -i \sigma_y
        \end{array}\right\}, 
\end{equation}
 whereas,the unitaries
$\{U^z\}$ chosen by Bob2 are $U^0 = I$ and $U^1 = \sigma_x$. After the encoding the shared GHZ state becomes 
\begin{equation}\label{eq:encoding}
\mathbb{I}\otimes U^{xy}\otimes U^{z} \ket{GHZ} = \frac{1}{\sqrt{2}} \sum_{l = 0}^1 (-1)^{l.y}\ket{l,x\oplus l,z \oplus l} = \ket{G^0(x,y,z)},
\end{equation} 
It is clear that the states, belonging to the set  $\{\ket{G^0(x,y,z)}\}_{x,y,z = 0}^1$, are mutually orthogonal and they form a complete set of basis vectors in the Hilbert space of three qubits $({\cal H}^2)^{\otimes 3}$. 
Alice can easily decode the messages by  performing a global  measurement in the basis
\begin{equation}\label{Eq:GHZ_basis}
{\cal B}_{GHZ} = \big\{\ket{G^0(x,y,z)}\bra{G^0(x,y,x)}, ~~x,y,z \in \{0,1\}\big\}.
\end{equation}

Moreover,  Bob2, can also choose other two Pauli matrices  to encode his bit $z \in \{0,1\}$,  i.e.,  $U^0 = \sigma_z$ and $U^1 = -i \sigma_y$, and the protocol will work equally well, as Eq. (\ref{eq:encoding}) then becomes
\begin{equation}\label{eq:encoding2}
\mathbb{I}\otimes U^{xy}\otimes U^{z} \ket{GHZ} = \frac{1}{\sqrt{2}} \sum_{l = 0}^1 (-1)^{l.(y\oplus 1)}\ket{l,x\oplus l,z \oplus l} = \ket{G^1(x,y,z)}.
\end{equation} 
This new set  $\big\{\ket{G^1(x,y,z)}\bra{G^1(x,y,z)}\big\}_{x,y,z = 0}^1 =: { {\cal B}}_{GHZ}$ also forms the same basis in the Hilbert space of 3-qubit system with a new labeling of its elements. One can easily check that the elements are related by $\ket{G^1(x,y,z)} = \ket{G^0(x,y\oplus 1,z)}$, $\forall x,y,z \in \{0,1\}$.

From Eqs. (\ref{eq:encoding}) and (\ref{eq:encoding2}), it should be clear that we can symmetrize the encoding operations performed by both the senders. We can do so by allowing both of them to encode two bits of information. In that scenario Bob1 encodes $(x,y)$ by applying a unitary $U^{x,y}$ while Bob2 encodes $(z,s)$ by applying $U^{z,s}$, 
  as in Eq. (\ref{Eq:encode_uni}) resulting in the following shared state
\begin{equation}\label{eq:encoding3}
\mathbb{I}\otimes U^{xy}\otimes U^{zs} \ket{GHZ} = \frac{1}{\sqrt{2}} \sum_{l = 0}^1 (-1)^{l.(y\oplus s)}\ket{l,x\oplus l,z \oplus l} = \ket{G^s(x,y,z)} = \ket{G^y(x,s,z)}.
\end{equation} 
The above equations explicitly show that the second bit $``y"$ of Bob1 and $``s"$ of Bob2 are treated equally. Any one of the bits $``y"$ or $``s"$ needs to be shared publicly for the proper execution of the secure dense coding protocol. Whoever needs to send more information can keep his second bit while the other sender can disclose his bit. The choice can be decided by mutual consent. Moreover, leveraging time-sharing, Bob1 can disclose his $``y"$ bit half of the time and similarly Bob2 his $``s"$ if both of them need to communicate an equal amount of information. 

In the remaining part of our paper, we assume, without any loss of generality, that it is Bob2 who discloses his second bit $``s"$ publicly. We call it an auxiliary bit. After Alice learns the value of the auxiliary bit she performs the same ${\cal B}_{GHZ}$ basis measurement, but her decoding procedure, i.e., the identification of the messages is chosen from Table \ref{Table:Secure_DC}, if $s = 0$ and from Table \ref{Table:Secure_DC2}, if $s = 1$.





\begin{figure}[t]
\centering
 \includegraphics[width=1\columnwidth,keepaspectratio,angle=0]{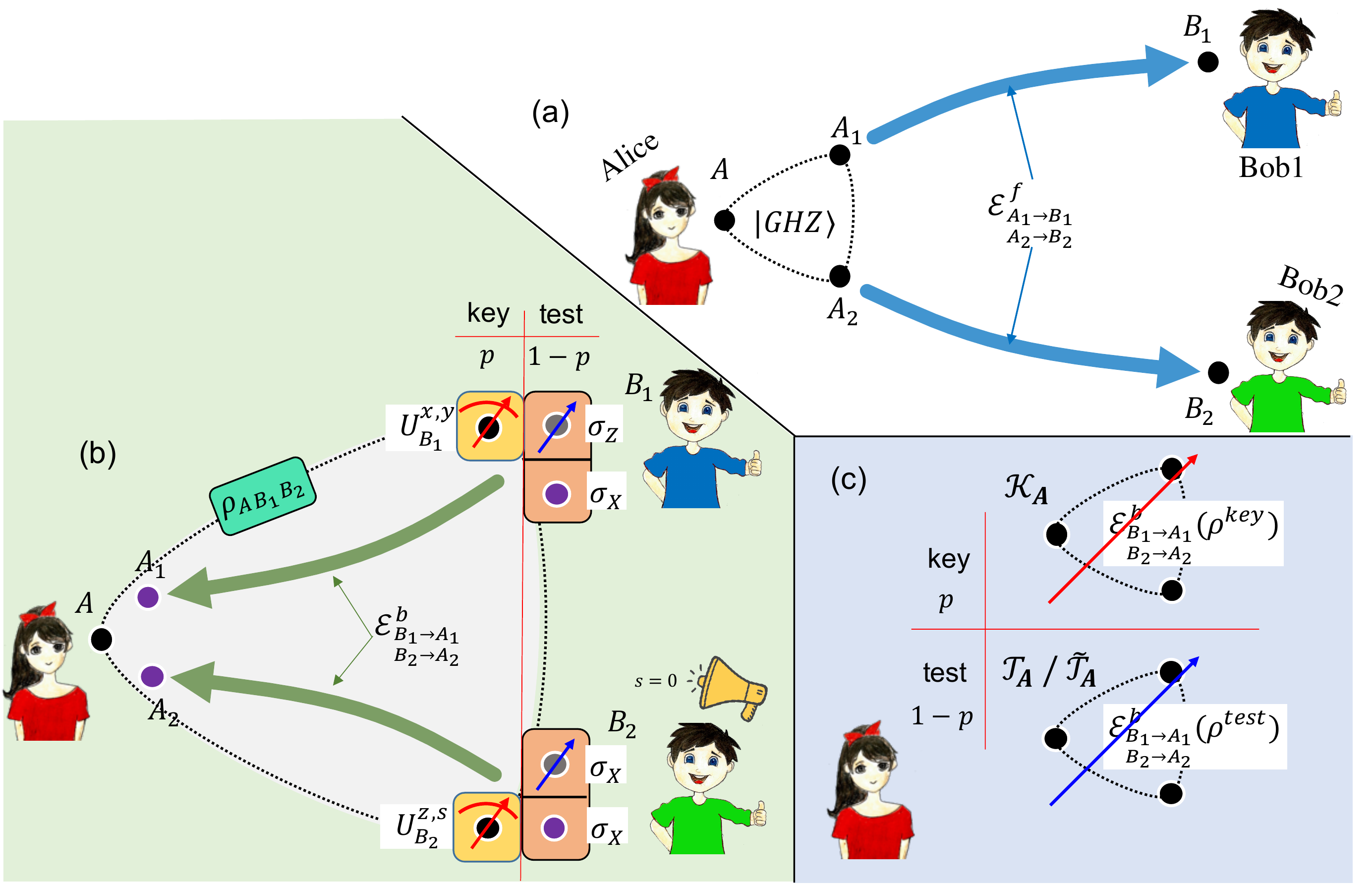}
\caption{Detailed schematic  diagram for 2-1 secure dense coding protocol.  In panel $(a)$, Alice prepares a three-qubit $\ket{GHZ}_{AA_1A_2}$ state. She keeps one qubit (part $A$), with her and sends the other two qubits, $A_1$ and $A_2$ to the two senders Bob1 and Bob2,  by using a quantum channel ${\cal E}^f$. The channel has been shown by the blue arrows in panel (a). After receiving the noisy state $\rho_{AB_1B_2}$, both the senders perform their local  encoding operations, $\{U^{\alpha,\beta}_{B_i}\}_{\alpha,\beta = 0}^1$ for $i = 0,1$, chosen uniformly from the three Pauli matrices and identity operators, as part of their key generation run which they perform with  probability $p \approx 1$.  With the remaining probability, $1 - p$, both of them perform the test run, during which they locally measure their shared part of $\rho_{AB_1B_2}$  in the eigenbasis of $\sigma_z $, $\{\ket{0}, \ket{1}\}$, for Bob1 and in the eigenbasis of $\sigma_x $, $\{\ket{+}, \ket{-}\}$, for Bob2.
They also locally prepare states randomly chosen from the eigenbasis of $\sigma_x$, i.e.,  $\{\ket{+}, \ket{-}\}$.  Bob2 announces the value of his auxiliary bit $s$ publicly, irrespective of his test and key generation run.
This part of the protocol has been depicted in panel (b). 
In the next step, the senders send their states back to Alice by using a backward transmission quantum channel ${\cal E}^b$ as depicted in (b). 
After getting the shared state back in her laboratory, Alice performs a  measurement on the entire three qubit states, with probability $p$, denoted by the ${\cal K}_{AA_1A_2} = \{\ket{G^s(i,j,k)}\bra{G^s(i,j,k)}_{AA_1A_2}\}_{i,j,k= 0}^1$, to decode the secure messages she got from both senders. Otherwise, with probability $1-p$ she measures either in the basis of ${\cal T}_{AA_1A_2} = \{ \hat{T}^{i,j}_{AA_1A_2}\}_{i,j = 0}^1$ or $\tilde{\cal T}_{AA_1A_2} = \{(\hat{\tilde{T}}^{k}_s)_{AA_1A_2}\}_{k = 0}^1$ as part of her test run, which is depicted in panel (c), with $\mathbf{A} = AA_1A_2$. }
\label{fig:SDC-schematic}
 \end{figure}

  \section{Multiparty Secure key distribution protocol based on multiparty dense coding protocol}
\label{sec:secure-dc}

This section describes a secure key distribution protocol based on the multipartite dense coding discussed earlier. Here, the aim is for 
Alice and  Bob1 and Alice and Bob2 to share independent secret keys.
We refer to this protocol as {\it $2$-$1$ secure dense coding protocol}. The secret key shared between Alice and Bob1 is not correlated with the key shared between Alice and Bob2. Thus the considered scenario is different from the quantum conference key agreement \cite{CKA_review}.

The protocol we are going to present is a generalization of the well-known ping-pong protocol and the dense coding-based QKD protocol proposed by Beaudry {\it et al.} \cite{Beaudry2013}.

In our protocol $\ket{GHZ}_{AB_1B_2}$ state  is shared between  Alice, Bob1 and Bob2, as depicted in figure  \ref{fig:SDC-schematic}. After receiving the respective parts of the shared GHZ, both the senders (Bob1 and Bob2) encode their secret key bits by performing unitary operations (see Sec \ref{sec:DC_def}). This step we shall call the {\it DC key generation run by senders}. The honest parties choose this action with probability $ p \approx 1$. 
 After the encoding procedure, both the senders send their qubits back to Alice via a quantum channel. Alice then decodes the secret key bits by performing ${\cal B}_{GHZ}$ basis measurements upon the suitable choice of the Table dependent on the auxiliary bit.
 
In the ideal situation, when the shared state has not been affected by any noise, Alice can share two bits of secure keys with one sender and an additional one bit with the other in each run of the protocol. 
However, in general, one can not assume the ideal channel between senders and a receiver. For example, a malicious eavesdropper, Eve, can try to intercept the information. Eve can attack the channel in two possible ways, either at the time of sharing the $\ket{GHZ}$ state or at the time of transmitting its encoded subsystems back. To assure the security of this protocol, the honest parties need to know how close the shared state is to the pure $\ket{GHZ}$, which determines how correlated the eavesdropper can be with the honest parties. Hence, a test run needs to been carried out with probability $1-p$ to detect the presence of an eavesdropper. The test run will be performed by each of the parties independently, with a small but non-zero probability   $1 - p$.
 
 The steps of  $2$-$1$ secure dense coding (DC) protocol are as follows.

\begin{itemize}
     \item {\bf State preparation:} Alice prepares a three-qubit GHZ state, $\ket{GHZ}_{AA_1A_2}$. She keeps one qubit (part $A$) in her  quantum memory (or ``cloud"), and sends the other two qubits to Bob1 and Bob2.
     
    \item {\bf State transfer:} A forward quantum channel ${\cal E}^f_{\substack{A_1\rightarrow B_1 \\ A_2\rightarrow B_2}}$ is used to transfer the sub-system 
     $A_1A_2$ of the GHZ state to the senders Bob1 and Bob2 respectively. The channel can be noisy resulting in a shared mixed state 
     \begin{equation}
         \rho_{AB_1B_2} = {\cal E}^f_{\substack{A_1\rightarrow B_1 \\ A_2\rightarrow B_2}}\left(\ket{GHZ}\bra{GHZ}_{AA_1A_2}\right),
     \end{equation}
      see figure \ref{fig:SDC-schematic} (a).  
    \item {\bf Key generation run:} Both the senders, with probability $p \approx 1$, initiate the key generation procedure independently,  by performing the following encoding operations, (see figure \ref{fig:SDC-schematic} (b)):
    \begin{itemize}
        \item  {\bf Encoding1:} 
        Bob1 uniformly chooses two bits of key $(x,y)$,  
        which are generated from a trusted uniform random number generator. He encodes his two bits $(x,y)$, by performing unitary encoding $U_{B_1}^{x,y}$, given in Eq. (\ref{Eq:encode_uni}), on his part of the shared state $\rho_{AB_1B_2}$ (see  Sec. \ref{sec:DC_def}), where $x,y \in \{0,1\}$.   
        \item {\bf Encoding2:} Bob2 applies uniformly at random one of the four unitary operators $U_{B_2}^{z,s}$, as given in Eq. (\ref{Eq:encode_uni}) on his part of the quantum state. $z$ denotes his secret key bit and $s$ the auxiliary bit, where $z,s \in \{0,1\}$. 
        \item {\bf Transferring states back:} Both the senders  send their respective encoded part back to Alice, by using backward  quantum transmission channel ${\cal E}^b_{\substack{B_1\rightarrow A_1 \\ B_2\rightarrow A_2}}$. Alice receives the state 
        \begin{equation}
            {\cal E}^b_{\substack{B_1\rightarrow A_1 \\ B_2\rightarrow A_2}} \big(U_{B_1}^{x,y} \otimes U_{B_2}^{z,s} \rho_{AB_1B_2} U_{B_1}^{x,y\dagger} \otimes U_{B_2}^{z,s \dagger} \big).
        \end{equation}
        Eve can intercept this channel to learn about the secret key bits of Alice and Bobs.
        \item {\bf Bit announcement:} Bob2 publicly announces the value of his auxiliary bit $s$.
    \end{itemize}
    \item {\bf Test run: } To detect the presence of eavesdropper, both the senders perform test run independently, with probability $1 - p$ in the following sequence of actions:
\begin{itemize}
\item Bob1 performs a measurement in the eigenbasis of $\sigma_z$, i.e., in the $\{\ket{0}, \ket{1}\} := \{\ket{\alpha}_{\alpha = 0}^1\}$ basis on his part of the shared state and prepares a state taken randomly from the set of the eigenvectors of $\sigma_x$ i.e., from $\{\ket{+}, \ket{-}\} := \{\ket{\beta_\vdash}\}_{\beta = 0}^1$.
\item Bob2 performs a measurement in the eigenbasis of $\sigma_x$ and prepare a state also in the eigenbasis of $\sigma_x$. If his measurement outcome and preparation are same then his auxiliary bit is $s = 0$ and if they differ then he discloses his auxiliary bit to be $s = 1$. 
\item Finally both the senders send their prepared states back to Alice via the backward quantum transmission channel ${\cal E}^b_{\substack{B_1\rightarrow A_1 \\ B_2\rightarrow A_2}}$. 
\end{itemize}
\item  After receiving the states from both Bobs, Alice performs a decoding operation of the key generation run with probability $p$ for each of the senders and a test run with small but non-zero probability $1 - p$ in order to detect the presence of an eavesdropper. The measurement is carried out jointly on all three qubits, as depicted in figure \ref{fig:SDC-schematic} (c).
\begin{itemize}
    \item {\bf With Bob1 :} To share the local keys securely with Bob1, Alice performs key generation run - decoding of the keys sent by Bob1 - with probability $p \approx 1$. With the remaining $1 - p$ she performs the test run.
    \begin{itemize}
        \item {\bf Decoding procedure:} To identify the raw keys shared with Bob1, she performs a rank two projective measurements  ${\cal K}^2_{AA_1A_2} = \{\bar{K}^{i,j}_s\}_{i,j = 0}^1$, which depends on the auxiliary bit $s$, where $\tilde{K}^{i,j}_s = \sum_{k' = 0}^1 \ket{G^s(i,j,k')}\bra{G^s(i,j,k')}_{AA_1A_2}$.
        \item {\bf Testing:} In the test run, she performs a joint measurement in the product  eigenbasis of $(\sigma_z)_A \otimes (\sigma_x)_{A_1} \otimes \mathbb{I}_{A_2}$, denoted as ${\cal T}_{AA_1A_2} = \{\hat{T}^{i,j}_{AA_1A_2}\}_{i,j = 0}^1$, where $\hat{T}^{i,j}_{AA_1A_2} =  \ket{i}\bra{i}_A \otimes \ket{j_\vdash}\bra{j_\vdash}_{A_1} \otimes \mathbb{I}_{A_2}$.
    \end{itemize}
    \item {\bf With Bob2 :}  As for Bob1, Alice applies the key generation and test protocols with probability $p$ and $1 - p$ respectively.
    \begin{itemize}
        \item  {\bf Decoding procedure:} She applies a rank four projective measurements  ${\cal K}^4_{AA_1A_2} = \{\tilde{K}^{k}_s\}_{k = 0}^1$ to obtain the key bit of Bob2, where $\tilde{K}^{k}_s = \sum_{i',j' = 0}^1 \ket{G^s(i',j',k)}\bra{G^s(i',j',k)}_{AA_1A_2}$.
        \item {\bf Testing:} Test run invloves the measurements in the eigenbasis of $\sigma_x$ performed only on the subsystem $A_2$, denoted as $\tilde{\cal T}_{AA_1A_2} = \{(\hat{\tilde{T}}^{k}_s)_{AA_1A_2}\}_{k = 0}^1$, where $\{(\hat{\tilde{T}}^{k}_s)_{AA_1A_2}\}_{k = 0}^1 = \mathbb{I}_{A} \otimes \mathbb{I}_{A_1}\otimes \ket{(k\oplus s)_\vdash}\bra{(k \oplus s)_\vdash}_{A_2} $.
    \end{itemize}
\end{itemize}
\item {\bf Classical post-processing: } The classical post-processing consists of two steps. First, Alice runs data compression individually with two Bobs \cite{RR}. Then they run a privacy amplification protocol by applying two-universal hash functions \cite{tomamichel2011leftover}. The respective protocols are described in more details in Remarks \ref{rem:datacompression} and  \ref{rem:privacyamplification}.
\end{itemize}

After the test and the key generation runs, all parties disclose the results of their measurements in the test run and some parts of the key generation run. From the correlation of their measurement outcomes, they estimate certain statistics. If the lower bound for the key rate based on these statistics is not positive, they abort the protocol. 
However, if the rate is positive, they perform classical post-processing on the remaining outcomes of the key generation runs. Explicit form of the acceptance/abortion condition
is provided in Remark \ref{rem:acc_abort}.
{\remark[Conditions for acceptance/aborting]
\label{rem:acc_abort}
After performing both the key generation runs and the test runs, all the parties disclose their measurement outcomes from each test run. Based on these measurement outcomes, all parties can estimate the conditional probabilities given in Eqs. (\ref{Eq:kappa_def})  and (\ref{eq:varrho_def}). They further estimate the key rate by computing the conditional Shannon entropies $( 1/4   \sum_{z,s} H(T_{A_1}|T_{B_1})_{{\sigma^s}(z)})$ and $(1/8 \sum_{xys} H(T_{A_2}|T_{B_2})_{\sigma^s(xy)})$. They abort the protocol if the estimated lower bounds for the local key rates are negative. If both of them are positive, they perform data compression, and privacy amplification adapted to the expected value of the key rate.
}

{\remark \label{rem:datacompression} Data compression protocol makes use of a family of two-universal hash functions $f: \mathcal{X} \rightarrow \{  0,1\}^n  $ i.e. a family of hash functions s.t. for a randomly drawn $f$, the probability of collision $f(x)=f(x')$, for two different inputs $x \neq x'$ is at most $1/2^n$. In a data compression step between Alice and Bob1, Alice chooses a random two-universal hash function $f_1$, applies it to its raw key and sends the result $c_1$ to Bob1. He then applies a pretty good measurement as in Lemma 1 in \cite{RR}.
The protocol is equivalent for Alice and Bob2. 
}
{\remark \label{rem:privacyamplification} In the privacy amplification step, Alice randomly chooses a two-universal hash function $g_1$, announces it publicly, and both Alice and Bob1 apply it to their common string. Similarly, she and Bob2 apply a two-universal hash function $g_2$. Thus, by the Leftover Hash Lemma, the outputs are almost uniformly random against quantum adversaries.
}

{\remark \label{rem:sampling} Note that all of the three parties choose the key generation run randomly with probability $p \approx 1$, and test run with probability $1 - p$. To agree on a random subset of runs for the test, they can use a certain amount of secure key, which they also need for authentication. It is
important to note that they require only $O(\sqrt{n})$ of such key since a random sample of size $\sqrt{n}$ represents the sampled set
with exponential precision \cite{Lo_Chau_Ardechali}. In this way, they can avoid mismatches in the choice of runs for testing to agree on a common subset of all the runs.}

 \section{Purification protocol}
 \label{sec:purification_protocol}
 
In this section, we will show that the encoding operation in the key generation run and the measurements performed in the test run by Bob1 and Bob2 can be purified to a joint von-Neumann measurement on the shared state and some suitably chosen auxiliary state. The {\it purified $2$-$1$ secure dense coding protocol} differs from the $2$-$1$ SDC protocol only in the encoding performed by the senders. In the key generation run, the senders replace Pauli rotation by teleportation of their state via the singlet state. In this way, the output gets rotated by the Pauli operations. Similarly, in the test run, they measure half of the singlet state instead of preparing state in some basis. Since these operations are done on the Bobs' sites, Eve can not tell apart the $2$-$1$ SDC protocol from the purified one. Indeed by assumption, she can not access the sites of Bob1 and Bob2, while their actions result in identical states as in standard $2$-$1$ SDC protocol. We explicitly show the purified encodings below. 

First note that in general the state shared between Alice, Bob1 and Bob2 does not have to be a pure multipartite entangled $\ket{GHZ}$ state. Instead, we write it is as a mixed state $\rho_{AB_1B_2} = {\cal E}^f_{\substack{A_1 \rightarrow B_1\\ A_2\rightarrow B_2}}(\ket{GHZ}\bra{GHZ}_{AA_1A_2})$. We assume that the eavesdropper possess the purificiation of this system, so that the joint state of Alice, Bob1, Bob2 and Eve is $\ket{\psi}_{AB_1B_2E}$. Here $E$ includes a joint presence of the eavesdropper in the key distribution protocol between Alice and Bob1 as well Alice and Bob2 so that $E = E_1E_2$
with $\rho_{AB_1B_2} = \text{tr}_{E}(\ket{\psi}\bra{\psi}_{AB_1B_2E})$.
\color{black}

The encoding operations performed by both senders in their respective parts of the shared $\rho_{AB_1B_2}$ are a completely positive trace preserving (CPTP) maps $\{\Theta^{x,y}_{B_1}\}$ and  $\{\Theta^{z,s}_{B_2}\}$, i.e., 
\begin{eqnarray}\label{eq:Uni_CPTP}
\Theta^{x,y}_{B_1 }\otimes \Theta^{z,s}_{B_2}(\rho_{AB_1B_2}) 
&=& \big(U_{B_1}^{x,y} \otimes U_{B_2}^{z,s}\big) \rho_{AB_1B_2} \big(U_{B_1}^{x,y\dagger} \otimes U_{B_2}^{z,s \dagger}\big).
\end{eqnarray}

\begin{figure}[h]
\centering
 \includegraphics[width=1\columnwidth,keepaspectratio,angle=0]{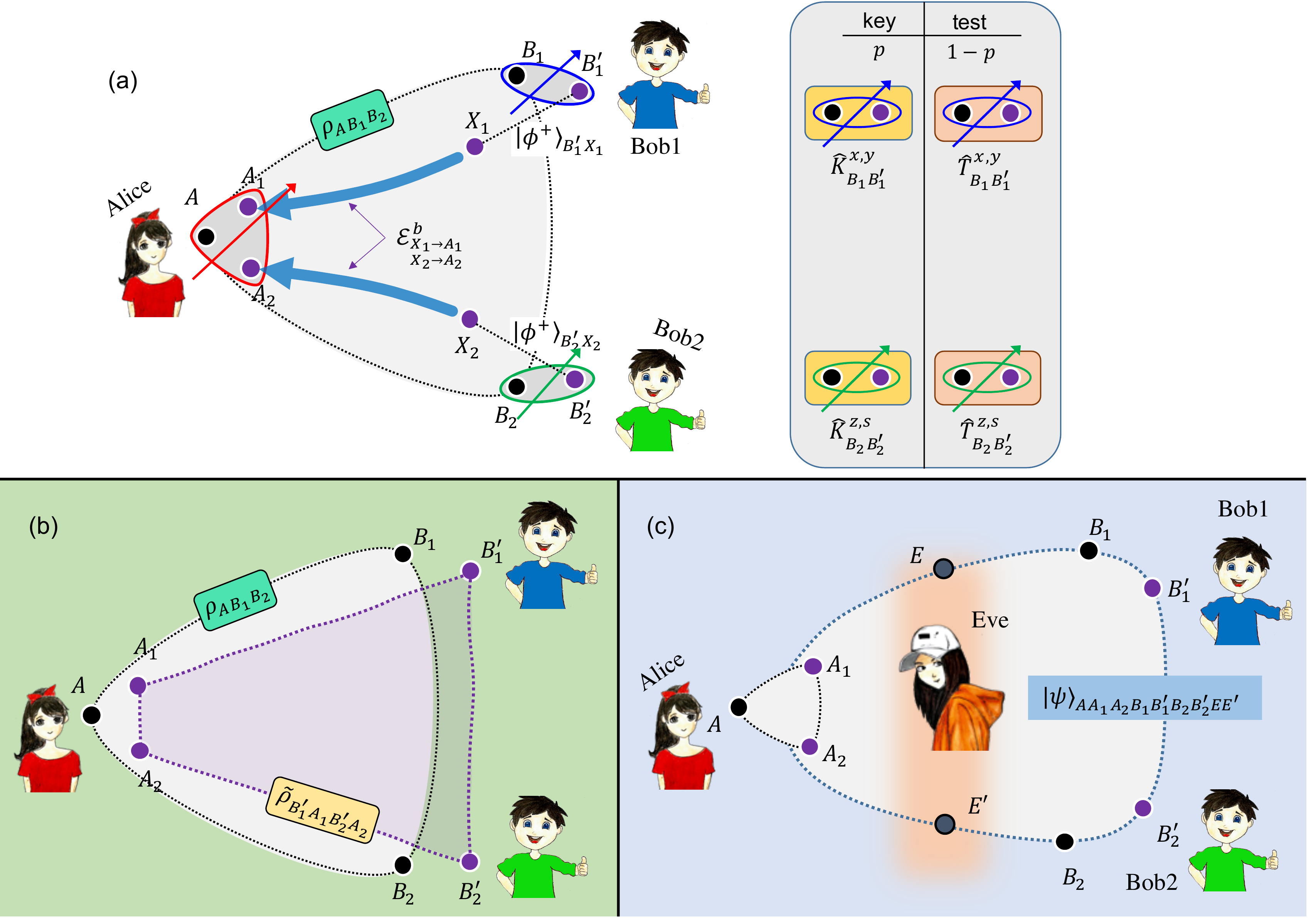}
\caption{Schematic diagram of the {\it purified $2$-$1$ secure dense coding protocol}. The unitary encodings $\{U^{\alpha,\beta}_{B_i}\}_{\alpha,\beta = 0}^1$ for $i = 0,1$ that both senders have performed on their respective parts of the shared noisy GHZ state $\rho_{AB_1B_2}$ can be purified to a joint von Neumann measurement $\{K^{\alpha,\beta}_{B_1B'_i}\}_{\alpha,\beta = 0}^1$, where $K^{\alpha,\beta}_{B_iB'_i} = \ket{B(\alpha,\beta)}\bra{B(\alpha,\beta)}_{B_iB'_i}$,  on the shared part of $\rho_{AB_1B_2}$ and a half of the auxiliary Bell state $\ket{\phi^+}_{B'_iX_i}$ for $i = 1,2$.
The measurement and preparation procedure  performed in the test run can also be purified in a similar manner wherein the joint measurements are $\hat{T}^{x,y}_{B_1B'_1}$  and $\hat{T}^{z,s}_{B_2B'_2}$
 given in Sec. \ref{sec:ratesA1B1} and \ref{sec:ratesA2B2}.
After the measurement, the other subsystem $X_i$ of the auxiliary Bell state is transferred back to Alice via the quantum channel ${\cal E}^b$. The purification protocol is shown in panel (a).  The effect of noise in the forward and backward channels is presented in panel (b). Noise in the backward channel affects the subsystem $X_i$ of the  state $\ket{\phi^+}_{B'_iX_i}$, for $i = 1,2$,  and transforms it to a mixed state $\tilde{\rho}_{B'_1A_1B'_2A_2} $, as shown in Eq. (\ref{eq:noiseonpuri}). We have introduce the eavesdropper in the system by giving her full access to the additional interfaces of the purification $\ket{\Psi}_{AA_1A_2B_1B'_1B_2B'_2EE'}$ of the product of mixed state $\rho_{AB_1B_2} \otimes \tilde{\rho}_{B'_1A_1B'_2A_2}$. The pure state shared among all the parties has been depicted in panel (c).
}
\label{fig:purification}
 \end{figure}

In our key generation run, the unitary rotations have been chosen from the three Pauli matrices $\{\sigma_x, \sigma_y, \sigma_z\}$ along with the identity operator $\mathbb{I}_2$, in the Hilbert space of dimension $2$.
Hence, using the teleportation protocol \cite{BBCJPW93} of an unknown quantum state, one can easily find that, the CPTP maps can be purified to a joint von-Neumann measurements $\{\hat{K}^{x,y}_{B_1B'_1}\}_{x,y = 0}^1$ and $\{\hat{K}^{z,s}_{B_2B'_2}\}_{z,s}^1$,
\footnote{This von-Neumann measurement operator should be a member of the complete set of measurement operators $\{\hat{M}^i\}_{i=1}^n$, such that it is projective $ (\hat{M}^{i})^2 = \hat{M}^{i}$ and satisfies the completeness relation $\sum_{i = 1}^n \hat{M}^{i} = \mathbb{I}$.}
performed on one subsystem of the shared state, $B_j$ and on another subsystem $B'_j$, of an auxiliary Bell state, $\ket{\phi^+}_{B'_jX_j} = \frac{1}{\sqrt{2}}(\ket{00} + \ket{11})_{B'_jX_j}$ for both  $j = 1,2$\footnote{
It has been proved in Ref. \cite{Beaudry2013}, that  an arbitrary CPTP map $\{\Xi_B^x\}$, on a quantum state $\tau_{AB}$, can also be purified to joint POVM measurements $\{{\cal M}^x\}$ on  $\tau_{AB}$, and an auxiliary pure state $\ket{\psi}_{BB'}$, i.e., 
\begin{equation}
\Xi^{x}_{B}(\tau_{AB}) = n \times \text{tr}_{B'B''}({\cal M}^x_{B'B''}(\tau_{AB'} \otimes |\psi\rangle\langle\psi|_{B''B}){\cal M}^{x\dagger}_{B'B''}).
\end{equation}
Here, $n$ is the dimension of the Hilbert space of the system $B'B''$.
}, e.g.,
\begin{eqnarray}\label{eq:encode_puri}
\Theta^{x,y}_{X_1}\otimes \Theta^{z,s}_{X_2}\big(\rho_{AX_1X_2}\big) &=& 4^2 \times \tr_{B_1B'_1}\tr_{B_2B'_2}\big(\hat{K}^{x,y}_{B_1B'_1}\otimes \hat{K}^{z,s}_{B_2B'_2}(\rho_{AB_1B_2} \otimes |\phi^+\rangle\langle\phi^+|_{B'_1X_1} \nonumber \\
&& \hspace{5cm}\otimes |\phi^+\rangle\langle\phi^+|_{B'_2X_2}) \hat{K}^{{x,y}\dagger}_{B_1B'_1}\otimes  \hat{K}^{{z,s}\dagger}_{B_2B'_2}\big), 
\end{eqnarray} 
where $\hat{K}^{\alpha,\beta} = \ket{B(\alpha,\beta)}\bra{B(\alpha,\beta)},~\alpha,\beta \in\{0,1\}$, acting on the subsystem $B_jB'_j, ~j = 1,2$, are four orthogonal Bell states  given by 
\begin{equation}
    \ket{B(\alpha,\beta)} = \frac{1}{\sqrt{2}}(\ket{0,\alpha} + (-1)^\beta \ket{1,\alpha \oplus 1}) = \frac{1}{\sqrt{2}} \sum_{l = 0}^1 (-1)^{l.\beta} \ket{l, \alpha\oplus l}.
\end{equation}
On the left-hand side of Eq. (\ref{eq:encode_puri}), we have used the subscripts $X_1$ and $X_2$, to denote the additional subsystem of the Bell state. These subsystems initially belong to the senders, but after the measurements, they are transferred back to Alice, through the backward transmission channel  ${\cal E}^b_{\substack{X_1 \rightarrow A_1\\ X_2\rightarrow A_2}}$. \\
A detailed proof of Eq. (\ref{eq:encode_puri}) is given in the Appendix \ref{sec:proof_puri}.
Moreover, the total state now in part of Alice, reads as 
\begin{eqnarray}
&&{\cal E}^b_{\substack{X_1 \rightarrow A_1\\ X_2\rightarrow A_2}} \left(\Theta^{x,y}_{X_1}\otimes \Theta^{z,s}_{X_2}\big(\rho_{AX_1X_2}\big)\right) \nonumber \\
&=& 4^2 \times \tr_{B_1B'_1}\tr_{B_2B'_2}\bigg(\hat{K}^{x,y}_{B_1B'_1}\otimes \hat{K}^{z,s}_{B_2B'_2}\Big(\rho_{AB_1B_2}   \otimes {\cal E}^b_{\substack{X_1 \rightarrow A_1\\ X_2\rightarrow A_2}} \left(|\phi^+\rangle\langle\phi^+|_{B'_1X_1} \otimes |\phi^+\rangle\langle\phi^+|_{B'_2X_2}\right)\Big)\hat{K}^{{x,y}\dagger}_{B_1B'_1}\otimes  \hat{K}^{{z,s}\dagger}_{B_2B'_2}\bigg), \label{eq:puri_noise} ~~~~~\\
&=& 4^2 \times \tr_{B_1B'_1}\tr_{B_2B'_2}\Big(\hat{K}^{x,y}_{B_1B'_1}\otimes \hat{K}^{z,s}_{B_2B'_2}\left(\rho_{AB_1B_2}   \otimes \tilde{\rho}_{B'_1A_1B'_2A_2}\right)\hat{K}^{{x,y}\dagger}_{B_1B'_1}\otimes  \hat{K}^{{z,s}\dagger}_{B_2B'_2}\Big), \label{eq:puri_noise2} ~~~~~~~
\end{eqnarray}
where in Eq. (\ref{eq:puri_noise}) we have used the fact that the measurement and the noise are acting on two different subsystems and therefore they commute, while in Eq. (\ref{eq:puri_noise2}), we have assumed that
\begin{equation}\label{eq:noiseonpuri}
    \tilde{\rho}_{B'_1A_1B'_2A_2} = {\cal E}^b_{\substack{X_1 \rightarrow A_1\\ X_2\rightarrow A_2}} \left(|\phi^+\rangle\langle\phi^+|_{B'_1X_1} \otimes |\phi^+\rangle\langle\phi^+|_{B'_2X_2}\right)
\end{equation}
The presence of eavesdropper can be introduced in the backward transmission channel by giving her access to the purified system of the purification of the above state, i.e., $\ket{\tilde{\psi}}_{B'_1A_1B'_2A_2E'}$. 
Moreover, we can also assume a global presence of the eavesdropper who performs a joint attack on both transmission channels. This is introduced by considering the purification of the total state $\rho_{AB_1B_2}   \otimes \tilde{\rho}_{B'_1A_1B'_2A_2}$ as $\ket{\Psi}_{AA_1A_2B_1B'_1B_2B'_2EE'}$, whereas for independent attacks on individual channels we simply have $\ket{\Psi}_{AA_1A_2B_1B'_1B_2B'_2EE'} = \ket{\psi}_{AB_1B_2E}\otimes \ket{\tilde{\psi}}_{B'_1A_1B'_2A_2E'}$.

The purified version of the 
test run can be equivalently  expressed 
in terms of extended measurements on the shared state and the auxiliary system. 
In the proposed protocol, Bob1 performs projective measurements on his shared part of $\rho_{AB_1B_2}$ in the eigenbasis of $\sigma_z$, denoted by $\{\ket{\alpha}\}_{\alpha = 0}^1$, and prepares a pure state in the eigenbasis of $\sigma_x$, denoted by $\{\ket{\beta_\vdash}\}_{\beta =0}^1$.
Similarly, Bob2 performs projective $\sigma_x$ measurement on the shared $\rho_{AB_1B_2}$ and prepares a state in the eigenbasis of the same Pauli operator $\sigma_x$. 
This measurement and preparation protocol can be purified to a single joint von-Neumann measurement of both senders. Here Bob1's measurement and preparation procedure can be purified as a joint measurement $\{\hat{T}^{x,y}_{B_1B'_1} = \ket{xy_\vdash}\bra{xy_\vdash}_{B_1B'_1}\}$, on the shared GHZ and one half of the Bell state. Similarly,  for Bob2 the purified measurement is  $\{\hat{T}^{z,s}_{B_2B'_2} = \ket{z_\vdash,(z\oplus s)_\vdash}\bra{z_\vdash, (z \oplus s)_\vdash}_{B_2B_2'}\}$. Each of these   $\{\hat{T}^{x,y}\}_{x,y = 0}^1$ and $\{\hat{T}^{z,s}\}_{s,z = 0}^1$, are  a complete set of von-Neumann measurement operators.

Recall that both senders and the receiver perform the test run to detect the presence of an eavesdropper in the quantum channel from Alice to both Bobs. Since according to our  protocol 
both senders perform the test run and the key generation run completely randomly and independently there might be a situation when one sender performs key generation run and the other one test run or vice a versa. Test run by one sender and key generation by other happens with probability $p(1 - p)$, and the purification for such a situation reads as 
\begin{eqnarray}
&& \big((\mathbb{I}_{A} \otimes \ket{x}\bra{x}_{B_1} \otimes \mathbb{I}_{X_2} ) \Theta_{X_2}^{z,s}(\rho_{AB_1X_2})(\mathbb{I}_{A}\otimes \ket{x}\bra{x}_{B_1} \otimes \mathbb{I}_{X_2} )\big) \otimes \ket{y_\vdash}\bra{y_\vdash}_{X_1}, \label{eq:test_key_puri} \nonumber \\
&&= 4^2 \times \tr_{B_1B'_1}\tr_{B_2B'_2}\big(\hat{T}^{x,y}_{B_1B'_1}\otimes \hat{K}^{z,s}_{B_2B'_2}(\rho_{AB_1B_2} \otimes |\phi^+\rangle\langle\phi^+|_{B'_1X_1}  \otimes |\phi^+\rangle\langle\phi^+|_{B'_2X_2}) \hat{T}^{{x,y}\dagger}_{B_1B'_1}\otimes  \hat{K}^{{z,s}\dagger}_{B_2B'_2}\big), \label{eq:test_key_puri2}
\end{eqnarray}
and
\begin{eqnarray}
&& \big((\mathbb{I}_{A} \otimes \mathbb{I}_{X_1} \otimes \ket{z_\vdash}\bra{z_\vdash}_{B_2})\Theta_{X_1}^{x,y}(\rho_{AX_1B_2})(\mathbb{I}_{A} \otimes \mathbb{I}_{X_1}\otimes \ket{z_\vdash}\bra{z_\vdash}_{B_2}) \big) \otimes  \ket{(z\oplus s)_\vdash}\bra{(z \oplus s)_\vdash}_{X_2}, \label{eq:key_test_puri} \nonumber \\
&&= 4^2 \times \tr_{B_1B'_1}\tr_{B_2B'_2}\big(\hat{K}^{x,y}_{B_1B'_1}\otimes \hat{T}^{z,s}_{B_2B'_2}(\rho_{AB_1B_2} \otimes |\phi^+\rangle\langle\phi^+|_{B'_1X_1}  \otimes |\phi^+\rangle\langle\phi^+|_{B'_2X_2}) \hat{K}^{{x,y}\dagger}_{B_1B'_1}\otimes  \hat{T}^{{z,s}\dagger}_{B_2B'_2}\big). \label{eq:key_test_puri2}
\end{eqnarray}

With a very small probability of $\approx (1 - p)^2$, both the senders perform common test runs. They detect a global eavesdropper. The purification in that scenario can be expressed as 

\begin{eqnarray}
&& \big((\mathbb{I}_A \otimes \ket{x}\bra{x}_{B_1}\otimes \ket{z}\bra{z}_{B_2})\rho_{AB_1B_2}(\mathbb{I}_A \otimes \ket{x}\bra{x}_{B_1}\otimes \ket{z}\bra{z}_{B_2}) \big) \otimes \ket{y_\vdash}\bra{y_\vdash}_{X_1} \otimes \ket{s_\vdash}\bra{s_\vdash}_{X_2} \label{eq:test_puri} \nonumber \\
&=& \rho_{A}^{x,z} \otimes \ket{y_\vdash}\bra{y_\vdash}_{X_1} \otimes \ket{s_\vdash}\bra{s_\vdash}_{X_2}. \label{eq:test_puri3} \\
&=& 4^2 \times \tr_{B_1B'_1}\tr_{B_2B'_2}\big(\hat{T}^{x,y}_{B_1B'_1}\otimes \hat{T}^{z,s}_{B_2B'_2}(\rho_{AB_1B_2} \otimes |\phi^+\rangle\langle\phi^+|_{B'_1X_1}  \otimes |\phi^+\rangle\langle\phi^+|_{B'_2X_2}) \hat{T}^{{x,y}\dagger}_{B_1B'_1}\otimes  \hat{T}^{{z,s}\dagger}_{B_2B'_2}\big), \label{eq:test_puri2} 
\end{eqnarray}

 The state at the end of the test run is transformed back to Alice via the channel ${\cal E}^b_{\substack{X_1 \rightarrow A_1\\ X_2\rightarrow A_2}}$. As such, we can assume that all measurements performed by the senders Bob1 and Bob2 are on the joint pure state $\ket{\Psi}_{AA_1A_2B_1B'_1B_2B'_2EE'}$.

  \section{Derivation of secure key rate.}
\label{sec:derivation_of_key}
In this section, we will study the security of the purified protocol of key distribution between the Alice and Bob1 and Alice and Bob2. We imagine the presence of eavesdroppers Eve1 ($E_1$) and Eve2 ($E_2$), who are colluding and can intercept the forward channel ${\cal E}^f_{\substack{A_1 \rightarrow B_1 \\ A_2 \rightarrow B_2}}$, that connects Alice with Bob1 and Alice with Bob2, at the time of sharing the $\ket{GHZ}$. The same situation happens in the case of the backward channel ${\cal E}^b_{\substack{B_1 \rightarrow A_1 \\ B_2 \rightarrow A_2}}$,  connecting Bob1 with Alice  and Bob2 with Alice. To prove  the security of the proposed $2$-$1$ secure dense coding protocol, we provide a lower bound on the secret key rate that Alice can distribute collectively with Bob1 and Bob2 while running the purified protocol. The security of the $2$-$1$ secure dense coding protocol follows from the fact that these two protocols are identical from Eve's point of view, as we argue in Section \ref{sec:purification_protocol}.

Before we compute the key rate, we need to define the post-measured states, that both the senders and the receiver will share with malicious eavesdropper, after all the honest parties have performed key generation run. 
The entire security proof is based on the purified version of the $2$-$1$ SDC protocol.
In the key generation run,  Bob1 performs the projective measurement ${\cal K}_{B_1B'_1}  = \{K_{B_1B_1'}^{xy}\}_{x,y = 0}^1$, in four orthogonal Bell states,
where $\hat{K}^{\alpha,\beta} = \ket{B(\alpha,\beta)}\bra{B(\alpha,\beta)},~\alpha,\beta \in\{0,1\}$, with $\ket{B(\alpha,\beta)} = \frac{1}{\sqrt{2}} \sum_{l = 0}^1 (-1)^{l.\beta} \ket{l, \alpha\oplus l}$. Similarly 
 Bob2's projective measurement is  ${\cal K}_{B_2B'_2} =  \{K_{B_2B_2'}^{zs}\}_{z,s = 0}^1 $, in which $s$ denotes the auxiliary bit.
On the other hand, Alice performs measurement, 
 ${\cal K}_{AA_1A_2} = \{(\hat{K}^{i,j,k}_s)_{AA_1A_2}\}_{i,j,k = 0}^1, ~~s = 0,1 $,  whose elements are eight orthogonal rank one projectors $\{\ket{G^s(i,j,k)}\bra{G^s(i,j,k)}\}_{i,j,k = 0}^1$, for both $s = 0,1$, with $\ket{G^s(i,j,k)} = \frac{1}{\sqrt{2}} \sum_{l = 0}^1 (-1)^{l.(j\oplus s)}\ket{l,i\oplus l,k \oplus l}$, on the Hilbert space $({\cal H}^2)^{\otimes 3}$.   Note that Alice chooses her set of measurements based on the values of the auxiliary bit $s$ which is publicly announced by Bob2.
 
\begin{definition}{\normalfont(State at the end of the key generation run performed by both senders and Alice).}
 The classical-classical-quantum state distributed among all the parties in their common key generation run is given by
\begin{eqnarray}\label{Eq:kappa_def}
&&\kappa_{K_{A_1}K_{A_2}K_{B_1}K_{B_2}XEE'} \nonumber \\
&=& {\cal K}_{AA_1A_2} \otimes {\cal K}_{B_1B'_1} \otimes {\cal K}_{B_2B'_2} (\ket{\Psi}\bra{\Psi}_{AA_1A_2B_1B'_1B_2B'_2EE'}) \nonumber \\
&=& \sum_{ijk} \sum_{xy} \sum_{zs} ({\hat K}^{i,j,k}_s)_{AA_1A_2} \otimes  {\hat K}^{x,y}_{B_1B'_1} \otimes {\hat K}^{z,s}_{B_2B'_2} (\ket{\Psi}\bra{\Psi}_{AA_1A_2B_1B'_1B_2B'_2EE'})({\hat K}^{i,j,k}_s)_{AA_1A_2} \otimes  {\hat K}^{x,y}_{B_1B'_1} \otimes {\hat K}^{z,s}_{B_2B'_2}, \nonumber \\
&=& \sum_{ijk} \sum_{xy} \sum_{zs} p(ijk;xy;zs)  \ket{ij}\bra{ij}_{K_{A_1}} \otimes \ket{k}\bra{k}_{K_{A_2}} \otimes \ket{xy}\bra{xy}_{K_{B_1}} \otimes \ket{z}\bra{z}_{K_{B_2}} \otimes \ket{s}\bra{s}_X \otimes \rho_{EE'}^{ijk;xy;zs}
\end{eqnarray}
where, $\ket{\Psi}_{AA_1A_2B_1B'_1B_2B'_2EE'}$ is the joint pure state shared among all the parties, $p(ijk;xy;zs) = \tr(\varrho_{EE'}^{ijk;xy;zs})$, and  $\rho_{EE'}^{ijk;xy;zs} = \frac{1}{p(ijk;xy;zs)} \varrho_{EE'}^{ijk;xy;zs}$, where the non-normalized quantum state, $\varrho_{EE'}^{ijk;xy;zs}$ is given by 
\begin{eqnarray}\label{eq:varrho_def}
   \varrho_{EE'}^{ijk;xy;zs} =  \bra{G^s(i,j,k)}_{AA_1A_2} \otimes \bra{B(x,y)}_{B_1B'_1} \otimes \bra{B(z,s)}_{B_2B'_2} (\ket{\Psi}\bra{\Psi}_{AA_1A_2B_1B'_1B_2B'_2EE'}) \ket{G^s(i,j,k)}_{AA_1A_2} \nonumber \\
   \otimes \ket{B(x,y)}_{B_1B'_1} \otimes \ket{B(z,s)}_{B_2B'_2}.
\end{eqnarray}
\end{definition}
Here, $\ket{ij}\bra{ij}_{K_{A_1}} \otimes \ket{k}\bra{k}_{K_{A_2}} $ is the state of the register identifying the measurement outcomes $(i,j,k)$ of Alice. Similarly $\ket{xy}\bra{xy}_{K_{B_1}} \otimes \ket{z}\bra{z}_{K_{B_2}} \otimes \ket{s}\bra{s}_X$, represents   the same for Bob1 and Bob2 accounting for the outcomes $(x,y)$ and $(z,s)$ respectively. 
Moreover, the subscript $K_{Y_w}$, for $Y \in \{A,B\}$ and $w = 1,2$ represents the raw key string shared among the parties at the end of the common key generation run. Note that in our protocol, Alice's key string $K_{A_1}$ has the highest correlation with 
$K_{B_1}$ of Bob1, and $K_{A_2}$ is correlated with the key string of Bob2  $K_{B_2}$.

Note that the projective measurement of Bob2 yields a binary key string $(z,s)$ out of which only a single bit $z$ stays secret. The auxiliairy bit $s$ is disclosed publicly at the end of the protocol.
In Eq. (\ref{Eq:kappa_def}), we use subscript $X$ to denote the system of the register state, which keeps the information of $s$. Alice generates her key string $(i,j,k)$
depended on the value of
$s$, by performing the measurement ${\cal K}_{AA_1A_2}$, where the labeling of its elements is dependent on $s$.

 The probability of seeing outcome $(x,y;z,s)$, is 
\begin{equation}\label{Eq:probxyzs}
    p(x,y;z,s) = \sum_{ijk} p(ijk;xy;zs) = \frac{1}{16}.
\end{equation}
The proof of Eq. (\ref{Eq:probxyzs}), is given in Appendix \ref{sec:proof_equal_prob}. From this we find that 
the outcome of the auxiliary bit $s$, in each run is, $p(s) = \sum_{x,y,z = 0}^1 p(x,y;z,s) = \frac 12$.



We emphasise that in $2$-$1$ SDC protocol, the auxiliary bit ``$s$" plays a crucial role. After Bob2 discloses the value of $s$ publicly, Alice chooses her output bits, based on her measurement ${\cal K}_{AA_1A_2}$, according to the Table \ref{Table:Secure_DC} or \ref{Table:Secure_DC2}.
The sets of projective measurements ${\cal K}_{AA_1A_2} = \{(\hat{K}_s^{i,j,k})_{{AA_1A_2}}\}_{i,j,k = 0}^1$, are the same for both $s = 0$ and $s = 1$, but the labeling of the elements is linked to 
$\{(\hat{K}_1^{i,j,k})_{AA_1A_2} = (\hat{K}_0^{i,j \oplus 1,k})_{AA_1A_2}\}_{x,y,z = 0}^1$. 
Table \ref{Table:Secure_DC}, corresponds to the choice of Alice's measurements,  $\{(\hat{K}_0^{x,y,z})_{{AA_1A_2}}\}_{x,y,z = 0}^1$, (left column), and her identification of the encoded bits of Bob1 (middle column) and Bob2 (rightmost column), when $s = 0$. Table \ref{Table:Secure_DC2} describes the analogous case when $s = 1$.

Now Eq. (\ref{Eq:kappa_def}), can be expressed as $\kappa_{K_{A_1}K_{A_2}K_{B_1}K_{B_2}XEE'} = \sum_s p(s) \kappa^s_{K_{A_1}K_{A_2}K_{B_1}K_{B_2}EE'} \otimes \ket{s}\bra{s}_{X}$, 
and  hence,
the ccq state shared among the honest parties (See Eq. (\ref{Eq:kappa_def})), when the auxiliary bit  $s$ is disclosed, is
\begin{eqnarray}\label{Eq:kappa_def3}
 \kappa^s_{K_{A_1}K_{A_2}K_{B_1}K_{B_2}EE'} 
= \sum_{ijk} \sum_{xy} \sum_{z} p(ijk;xy;z|s)  \ket{ij}\bra{ij}_{K_{A_1}} \otimes \ket{k}\bra{k}_{K_{A_2}} \otimes \ket{xy}\bra{xy}_{K_{B_1}} \otimes \ket{z}\bra{z}_{K_{B_2}}  \otimes \rho_{EE'}^{ijk;xy;zs},
\end{eqnarray}
where $p(ijk;xy;z|s) = 2 p(ijk;xy;zs)$.

In order to calculate the key rates between the honest parties from now onwards, we consider the ccq state $\kappa^s_{K_{A_1}K_{A_2}K_{B_1}K_{B_2}EE'}$, averaged over $s$. 

\subsection{Proof of the security of the multiparty secure dense coding protocol }
Once, the honest parties, generate the  raw key strings among themselves, the ccq state, the three-party secure protocol consists of the following  two main steps: first running one-shot Renes-Renner protocol to establish the key between Alice and Bob1 and then running one-shot Renes-Renner protocol again for Alice and Bob2. Let us state the main theorem.

\begin{theorem}
	Let there be an input state $(\rho^{in}_{A_1B_1A_2B_2E})^{\otimes n}$ with $A_1,B_1,A_2,B_2$ being the classical registers, for sufficiently large $n$. Denote by ${\cal P}_{A_1B_1}$ a Renes-Renner protocol applied to subsystems  $A_1B_1$ with secrecy and correctness parameters less than $\epsilon_{1}$ and $\delta_{1}$ respectively and by 
	${\cal P}_{A_2B_2}$ a Renes-Renner protocol applied to  subsystems  $A_2B_2$ with correctness and secrecy parameter less than $\epsilon_{2}$ and $\delta_{2}$. Let us assume that
\begin{eqnarray}
I(A_1:B_1) \geq H(A_1) - \delta_1,  \label{eq:assum1}\\
I(A_2:B_2) \geq H(A_2) - \delta_2, \label{eq:assum2}\\
I(B_1:B_2)\leq\delta_3, \label{eq:assum3}
\end{eqnarray}
which the parties can estimate by running statistical tests. Then, for any  pair of the rates $r({\cal P}_i)$, satisfying 
	\begin{align}
	    r({\cal P}_1) \leq I(A_1:B_1) - I(A_1:E|B_2) -\delta_1 - 2\delta_2-\delta_3,  \\
	    r({\cal P}_2) \leq I(A_2:B_2)-I(A_2:E|B_1) - 2\delta_1 -\delta_2 - \delta_3,
	\end{align}
	the pair of protocols  outputs a state $\rho^{out}_{A_1B_1A_2B_2E}$ which satisfies:
	\begin{eqnarray}
	||\rho_{A_1B_1A_2B_2E}^{out} - K_{A_1B_1}\otimes K_{A_2B_2}\otimes\sigma_{EC_1C_2}||_1\leq 2(\epsilon_{1} +\epsilon_{2})
	\label{eq:claim2}
	\end{eqnarray}
	with $K_{A_iB_i}=(\frac{1}{2}({\ket{00}\bra{00}+\ket{11}\bra{11}}))^{\otimes (n \times r({\cal P}_i))}$.
	\label{thm:main}
\end{theorem}

\begin{proof}
Our figure of merit is the input state $(\rho^{in}_{A_1B_1A_2B_2E})$.
Applying the first protocol  ${\cal P}_1$ to a single copy of this state gives the output state
	\begin{eqnarray}
	||{\cal P}_1(\rho^{in}_{A_1B_1A_2B_2E}) - K^{(1)}_{A_1B_1}\otimes \rho_{A_2B_2EC_1} ||_1\leq 2\epsilon_1
	\label{eq:output1}
	\end{eqnarray}
with the rate
\begin{align}
\ell_{\mathrm{secr}}^{\epsilon_1}  (A_1:B_1|A_2 B_2 E) \geq  H_{\mathrm{min}}^{\epsilon_{3}'}(A_1|A_2B_2E)-H_{\mathrm{max}}^{\epsilon_3}(A_1|B_1)-4 \log\frac{1}{\epsilon_4}-3,
\end{align}
where $\epsilon_1 = \epsilon_3+\epsilon_3'+2\epsilon_4$.
We further note, that the protocol ${\cal P}_2$ acts on $\rho_{A_2B_2EC_1}$. Hence,
after the application of ${\cal P}_2$ to 
$\rho_{A_2B_2EC}$, its
rate reads     
\begin{align}
\ell_{\mathrm{secr}}^{\epsilon_2} (A_2:B_2|EA_1) \geq H_{\mathrm{min}}^{\epsilon_5'}(A_2|EC_1)-H_{\mathrm{max}}^{\epsilon_5}(A_2|B_2)-4 \log\frac{1}{\epsilon_6}-3
\end{align}
where $\epsilon_2=\epsilon_5+\epsilon_5'+2\epsilon_6$.
We further note that
\begin{align}
    H_{\mathrm{min}}^{\epsilon_5'}(A_2|EC_1)_{\rho_{A_2B_2EC_1}} \geq  H_{\mathrm{min}}^{\epsilon_5'}(A_2|EA_1)_{\rho^{in}_{A_2B_2EA_1}}
\end{align}
The above inequality follows from the fact that $\rho^{in}_{A_2B_2EA_1}$ is an extension of $\rho_{A_2B_2E}$ to a system $A_1$, and there exists a deterministic map 
on $A_1$ which outputs
$C_1$. Due to the data processing inequality (c.f. Theorem 5.7 in \cite{tomamichel2012framework}), $H_{\mathrm{min}}^{\epsilon_5'}(A_2|EA_1)$ can
only increase under this (local) map, hence
the inequality.

Now consider $n$ copies of the state $\rho^{in}_{A_1 A_2 B_2 E}$. The rate after applying the first protocol  ${\cal P}_1$ becomes
\begin{align}
\frac{1}{n} \ell_{\mathrm{secr}}^{\epsilon_1}  (A_1^n:B_1^n|A_2^n B_2 ^n E^n) \geq \frac{1}{n}\left( H_{\mathrm{min}}^{\epsilon_3'}(A_1^n|A_2^n B_2^n E^n )-H_{\mathrm{max}}^{\epsilon_3}(A_1^n|B_1 ^n)-4 \log\frac{1}{\epsilon_4}-3 \right)
\end{align}
Since the input state is i.i.d. in the asymptotic limit of many copies $n \to \infty$ we can apply the  fully quantum asymptotic equipartition property (see Result 6 in \cite{tomamichel2012framework}) to obtain
\begin{align}
r({\cal P}_1) \geq H(A_1|A_2B_2E)- H(A_1|B_1) = I(A_1:B_1)-I(A_1:EA_2B_2).
\end{align}

Similarly, the rate of  ${\cal P}_2$  in the many-copy case reads
\begin{align}
\frac{1}{n} \ell_{\mathrm{secr}}^{\epsilon_2}  (A_2^n:B_2^n|E^n A_1 ^n) \geq \frac{1}{n}\left( H_{\mathrm{min}}^{\epsilon_3'}(A_2^n|E^nA_1 ^n )-H_{\mathrm{max}}^{\epsilon_3}(A_2^n|B_2 ^n)-4 \log\frac{1}{\epsilon_4}-3 \right)
\end{align}
By corollary \ref{cor:iid}, the state $\rho^n_{A_1 A_2 B_2 E}$ is i.i.d 
$\rho^n_{A_1 A_2 B_2 E}= \rho^{\otimes n}_{A_1 A_2 B_2 E}$. and as before we can employ the  fully quantum asymptotic equipartition property to get

\begin{align}
r({\cal P}_2) \geq H(A_2|EA_1)- H(A_2|B_2) = I(A_2:B_2)-I(A_2:EA_1).
\end{align}

So far we have proven, that the rate of protocol ${\cal P}_1$ is $I(A_1:B_1)-I(A_1:EA_2B_2)$, and that of ${\cal P}_2$ reads $I(A_2:B_2) - I(A_2:EA_1)$. To finalize and lower bound these rates as claimed, we need to prove the following lemmas first.
\begin{lemma}
Let random variables $A_1,A_2,B_1,B_2$ satisfy the assumptions (\ref{eq:assum1}) to (\ref{eq:assum3}), 
then both $I(A_1:A_2B_2)$ and $I(A_1B_1:A_2)$ are bounded by $\delta_1+\delta_2 + \delta_3$, i.e., 
\begin{eqnarray}
    I(A_1:A_2B_2) \leq \delta_1+\delta_2 + \delta_3 \label{eq:IA1A2B2},\\
    I(A_1B_1:A_2) \leq \delta_1+\delta_2 + \delta_3 \label{eq:IA1B1A2}.
\end{eqnarray}
\label{lem:monogamy}
\end{lemma}
The idea of the above lemma is that the raw keys prepared by Bob1, $K_{B_1} \in \{x,y\}_{x,y = 0}^1$ and by Bob2, $K_{B_2} \in \{z\}_{z = 0}^1$ are by the virtue of the protocol de-correlated i.e. $p(x,y,z)=p(x,y)p(z)$. Hence, upon observing high correlation between $K_{A_1}$ and $K_{B_1}$ as well as $K_{A_2}$ and $K_{B_2}$, one can not
observe high correlations between $K_{A_1}$ and $K_{A_2}K_{B_2}$ and so between $K_{A_2}$ and $K_{A_1}K_{B_1}$ or else $K_{B_1}$ and $K_{B_2}$ would be correlated.

Before we prove the above lemma, we need to invoke another one:
\begin{lemma} For any random variables $S,T,U$ there is
\begin{equation}
I(S:T)+I(T:U)\leq I(S:U) + I(T:SU).
\end{equation}
\label{lem:useful}
\end{lemma}
\begin{proof}
For the proof of this lemma, see lemma 5 in the Supplemental Material of \cite{Grudka2014}. Finally, we are ready to show the proof of the Lemma \ref{lem:monogamy}.
\end{proof}

\begin{proof}{of Lemma \ref{lem:monogamy}}:
We first use the chain rule to observe that
\begin{eqnarray}
I(A_1:A_2B_2) = I(A_1:B_2) + I(A_1:A_2|B_2) \label{eq:first}
\end{eqnarray}
We focus on the first term in Eq. (\ref{eq:first}), and use Lemma \ref{lem:useful}.
\begin{eqnarray}
I(A_1:B_1)+I(A_1:B_2) \leq I(B_1:B_2) + I(A_1:B_1B_2)
\end{eqnarray}
We further note, that by assumption (\ref{eq:assum1}) $I(A_1:B_1) \geq H(A_1) -\delta_1$ while $I(B_1:B_2) \leq \delta_3$. We thus obtain
\begin{eqnarray}
H(A_1) -\delta_1 +I(A_1:B_2) \leq I(A_1:B_1B_2) +\delta_3 \leq H(A_1) +\delta_3
\end{eqnarray}
where in the last inequality we have used the fact that $I(X:Y)\leq \min \{H(X),H(Y)\}$.
Thus we obtain that 
\begin{equation}
I(A_1:B_2)\leq \delta_1 +\delta_3.
\label{eq:small_cross}
\end{equation}
Now in the second term of Eq. (\ref{eq:first}), we apply the sequence of (in)equalities  given by 
\begin{eqnarray}
I(A_1:A_2|B_2) &=& H(A_2|B_2) - H(A_1|A_2B_2) \\
&\leq & \delta_2 - H(A_1|A_2B_2) \leq  \delta_2,
\end{eqnarray}
where in the first inequality we have used (\ref{eq:assum2}) and the second inequality follows from the positivity of conditional  Shannon entropy $H(A_1|A_2B_2) \geq 0$.

Since the above argument is symmetric under the exchange of indices $1$ and $2$, this concludes the proof.
\end{proof}

To this end we first note, that for any ccq state $\rho_{A_1B_1A_2B_2E}$, with $A_1B_1A_2B_2$, being classical and $E$ quantum, satisfying  assumptions \eqref{eq:assum1}, \eqref{eq:assum2} and \eqref{eq:assum3}, we can employ Lemma \ref{lem:lb} to obtain 
$r({\cal P}_1) \geq I(A_1:B_1)-I(A_1:E|B_2) -\delta_1 - 2\delta_2-\delta_3$ and  $r({\cal P}_2)\geq I(A_2:B_2)-I(A_2:E|B_1) - 2\delta_1 - \delta_2 -\delta_3$.
We will now prove the following Lemma \ref{lem:lb}, 
with the help of Lemma  \ref{lem:monogamy}
\begin{lemma}
	\label{lem:lb} For any ccq state $\rho_{A_1B_1A_2B_2E}$ with $A_1B_1A_2B_2$ being classical and $E$ quantum, satisfying
 \begin{align}
    I(A_1:B_1)\geq H(A_1) - \delta_1, \label{eq4:as1}\\
        I(A_2:B_2) \geq H(A_2) -\delta_2, \label{eq4:as2}\\
         I(B_1:B_2) \leq \delta_3,\label{eq4:zero}
    \end{align}
there is 
\begin{eqnarray}
I(A_1:B_1) - I(A_1:A_2B_2E) &\geq&  
I(A_1:B_1) - I(A_1:E|B_2) - \delta_1 - 2\delta_2-\delta_3, \label{eq:ExpressionA1B1}
\end{eqnarray}
and 
\begin{eqnarray}
I(A_2:B_2) - I(A_2:EA_1)&\geq& I(A_2:B_2)-I(A_2:E|B_1) - 2\delta_1 - \delta_2 -\delta_3. \label{eq:ExpressionA2B2}
\end{eqnarray}
\end{lemma}

\begin{proof}
 First applying the chain rule of mutual information, in the l.h.s of (\ref{eq:ExpressionA1B1}),
	\begin{eqnarray}
	I(A_1:B_1) - I(A_1:A_2B_2E) = I(A_1:B_1) - [I(A_1:A_2B_2) + I(A_1:E|A_2B_2)]
	\end{eqnarray}
	
From Lemma \ref{lem:monogamy} the second term on the r.h.s. is greater or equal to $-\delta_1 - \delta_2-\delta_3$. We will lower bound the third term by
$-I(A_1:E|B_2)$. We have the following chain of (in)equalities upper bounding $I(A_1:E|A_2B_2)$.
\begin{eqnarray}
I(A_1:E|A_2B_2) &=& I(A_1:A_2B_2E) - I(A_1:A_2B_2) \nonumber \\
&\leq& I(A_1:A_2B_2E) -I(A_1:B_2)  \nonumber \\
&=&H(A_1)+ H(A_2B_2E) - H(A_1A_2B_2E) - I(A_1:B_2)  \nonumber\\
&\leq& H(A_1) + H(B_2E) + \delta_2 -H(A_1 B_2 E) -I(A_1 : B_2) \label{eq:hardest} \\
&=&I(A_1:E|B_2)+\delta_2
\end{eqnarray}
The first equality is just the chain rule. The first inequality follows from data processing i.e. $I(A_1:A_2B_2) \geq I(A_1:B_2)$.
The next equality follows from definition of the mutual information applied to the first term.
The inequality (\ref{eq:hardest}) is more involved. First we employ the data processing in $H(A_1 A_2 B_2 E) \geq H(A_1 B_2 E)$. Furthermore, to see that $H(A_2B_2E) \leq H(B_2E) + \delta_2 $ it is enough to show that
$H(A_2|B_2E)\leq \delta_2$. We know however by assumption (\ref{eq4:as2}) that 
\begin{eqnarray}
H(A_2) - H(A_2|B_2E) &=& I(A_2:B_2E) \nonumber \\
&\geq& I(A_2:B_2)\geq H(A_2) -\delta_2,
\end{eqnarray}
where we have used data processing inequality for quantum mutual information dropping $E$. Hence $H(A_2|B_2E) \leq \delta_2$ as claimed. The last equality follows from the definition of $I(A_1:E|B_2)$.
From this part of the proof, we conclude that
	\begin{eqnarray}
I(A_1:B_1) - I(A_1:A_2B_2E) \geq I(A_1:B_1) - I(A_1:E|B_2)  -\delta_1 -2\delta_2-\delta_3.
\end{eqnarray}

In order to prove (\ref{eq:ExpressionA2B2}), we first notice that from the data processing inequality, we obtain
\begin{align}
    I(A_2:B_2) - I(A_2:EA_1)\geq  I(A_2:B_2) - I(A_2:EA_1B_1).
\end{align}
Now we can apply Lemma \ref{lem:lb} with the indices 1 and 2 swapped, which we can do since the underlying assumption are symmetric under this swap. Then
\begin{align}
    I(A_2:B_2) - I(A_2:EA_1B_1)\geq  I(A_2:B_2) - I(A_2:E|B_1)-2\delta_1 -\delta_2-\delta_3
\end{align}
and completes the proof.
\end{proof}

Finally, we need to show the composition of the protocols and the secrecy parameters. After the application of 
${\cal P}_1$ to $(\rho^{in}_{A_1B_1A_2B_2E})^{\otimes n}$, by security of the first protocol, there is:
	\begin{eqnarray}
	||{\cal P}_1\left((\rho^{in}_{A_1B_1A_2B_2E})^{\otimes n}\right) - K_{A_1B_1}\otimes \rho_{A_2B_2EC_1} ||_1\leq 2\epsilon_1.
	\label{eq:output2}
	\end{eqnarray}
	Note here, that the key $K_{A_1B_1}$ is also decoupled from
	systems $A_2B_2$, not only $E$.
Protocol ${\cal P}_2$ acts on registers $A_2B_2$ and appends new register $C_2$ holding communication. Let us denote  $\rho_{A_1B_1A_2B_2EC_1}^{(n)}:={\cal P}_1\left((\rho^{in}_{A_1B_1A_2B_2E})^{\otimes n}\right)$. Since ${\cal P}_2$ acts only on systems $A_2B_2$ and appends communication $C_2$ to Eve it effectively has the input $\tr_{A_1B_1}\rho^{(n)}_{A_1B_1A_2B_2EC_1}=:\rho^{(n)}_{A_2B_2EC_1}$.
	\begin{eqnarray}
	||{\cal P}_2\left(\rho_{A_2B_2EC_1}^{(n)}\right) - K_{A_2B_2}\otimes \sigma_{EC_1C_2} ||_1\leq 2\epsilon_2
	\label{eq:output3}
	\end{eqnarray}
Further, for $\rho^{out}_{A_1B_1A_2B_2EC_1C_2}:={\cal P}_2({\cal P}_1\left((\rho^{in}_{A_1B_1A_2B_2E})^{\otimes n}\right))$ we have:
	\begin{align}
	&||\rho^{out}_{A_1B_1A_2B_2EC_1C_2}- K_{A_2B_2}\otimes K_{A_1B_1}\otimes \sigma_{EC_1C_2} ||_1 \nonumber \\&= ||\rho^{out}_{A_1B_1A_2B_2EC_1C_2}-  K_{A_1B_1}\otimes {\cal P}_2(\rho_{A_2B_2EC_1}^{(n)})+ K_{A_1B_1}\otimes {\cal P}_2(\rho_{A_2B_2EC_1}^{(n)})-K_{A_2B_2}\otimes K_{A_1B_1}\otimes\sigma_{EC_1C_2} ||_1 \nonumber \\
	&\leq||\rho^{out}_{A_1B_1A_2B_2EC_1C_2}-  K_{A_1B_1}\otimes {\cal P}_2(\rho_{A_2B_2EC_1}^{(n)})||_1 +|| K_{A_1B_1}\otimes {\cal P}_2(\rho_{A_2B_2EC_1}^{(n)})-K_{A_1B_1}\otimes K_{A_2B_2}\otimes\sigma_{EC_1C_2} ||_1 \nonumber \\
	&\leq 2\epsilon_1 +2\epsilon_2
	\end{align}
The first equality comes from adding and subtracting $K_{A_1B_1}\otimes {\cal P}_2(\rho_{A_2B_2EC_1}^{(n)})$,
the next line is just a triangle inequality,
The first in the last inequality is a combination of the Eq. \eqref{eq:output2} and the fact that the trace norm is non-increasing under local operations and public communication. The second term follows from the Eq. \eqref{eq:output3} and the multiplicativity of the trace norm under the tensor product. 
\end{proof}







\begin{observation}\label{obs:ncopies}
     For any state $\rho_{XY}^{\otimes n}$ and any channel $\Lambda: \mathcal{L}(\mathcal{H}_Y^{\otimes n}) \mapsto \mathcal{L}(\mathcal{H}_{\tilde{Y}})$ mapping $\Lambda(\rho_{Y}^{\otimes n})= \rho^{(n)}_{\tilde{Y}} $ there is 
     \begin{equation}
         \tr_{\rho^{(n)}_{\tilde{Y}}} ( \mathcal{I}_n \otimes \Lambda ) \rho_{XY}^{\otimes n} = \tr _{\rho_{Y}^{\otimes n}}\rho_{XY}^{\otimes n} = \rho_{X}^{\otimes n}
     \end{equation}
\end{observation}

\begin{proof}
Suppose the output state $ \rho_{\tilde{X}}$ was different from $\rho_{X}^{\otimes n}$. Then an observer $X$ having access to many copies of this state could differentiate between the two states, which would violate the no-signalling principle.
This is a direct application of the no-communication theorem, which can also be shown by writing the state as the convex sum with respect to the partition $XY$, expanding the channel in terms of its Kraus operators and invoking cyclicity of trace.
\end{proof}

\begin{corollary}
Consider an i.i.d. input state $\rho_{A_1B_1A_2B_2E}^{\otimes n}$. The following chain of physical operations results in
\begin{align}
\left(\rho^{in}_{A_1B_1A_2B_2E}\right)^{\otimes n} &\xrightarrow{\text{copy $A_1$} }
(\rho^{in}_{\hat{A}_1A_1B_1A_2B_2E})^{\otimes n} \xrightarrow{\text{apply $\mathcal{P}_1$}} \rho^{(n)}_{A_1K_{A_1}K_{B_1}A_2B_2EC_1}\\&\xrightarrow{\tr_{K_{A_1}K_{B_1}C_1}}\rho^{(n)}_{A_2B_2EA_1}
\end{align}
and $\rho^{(n)}_{A_2B_2EA_1}=\rho^{\otimes n}_{A_2B_2EA_1}$
\label{cor:iid}
\end{corollary}
\begin{proof}
Apply observation \ref{obs:ncopies} to the state $(\rho^{in}_{\hat{A}_1A_1B_1A_2B_2E})^{\otimes n}$ with $X=\hat{A}_1A_2B_2E$, $Y=A_1B_1$ and $\tilde{Y}=K_{A_1}K_{B_1}C_1$.
\end{proof}

\color{black}
\begin{corollary}
     The lower bound on the  local key rates, between Alice to Bob1 and Alice to Bob2, for the ccq state given in Eq. (\ref{Eq:kappa_def3}), shared   at the end of the common key generation run  is given by 
     \begin{eqnarray}
        r({\cal P}_1) \geq \frac 12 \sum_s \big( I(K_{A_1}:K_{B_1})_{\kappa^s} - I(K_{A_1}:E|K_{B_2})_{\kappa^s} - \delta_1(s) - 2\delta_2(s) \big) \label{eq:A1B1keyrate}\\
         r({\cal P}_2) \geq \frac 12 \sum_s \big( I(K_{A_2}:K_{B_2})_{\kappa^s} - I(K_{A_2}:E|K_{B_1})_{\kappa^s} - 2\delta_1(s) - \delta_2(s) \big) \label{eq:A2B2keyrate}
     \end{eqnarray}
     where by $\kappa^s$, in the subscript of the above inequalities we refer to $\kappa^s_{K_{A_1}K_{A_2}K_{B_1}K_{B_2}XEE'}$.
\end{corollary}
\begin{proof}
The proof of the corollary is immediate as Lemma \ref{lem:lb} 
is true for the ccq state $\kappa^s_{K_{A_1}K_{A_2}K_{B_1}K_{B_2}XEE'}$, for both $s = 0,1$. Averaging over $s$, and identifying the $s$ dependence of  $\delta_i, ~i = 1,2$ (given in  (\ref{eq:assum1}) and (\ref{eq:assum2})) and $\delta_3 = 0$.
\end{proof}

In the next section, we will simplify the lower bounds of the local key rates by using the entropic uncertainty relations.

 \subsection{Local key rates: Alice and Bob1}\label{sec:ratesA1B1}
 In this section, we will derive a computable expression of the lower bound on the secure key rate $r({\cal P}_1)$ between Alice and Bob1, given in  (\ref{eq:A1B1keyrate}). In our protocol, Bob1 prepares two bits of raw key $K_{B_1}\in \{x,y\}_{x,y = 0}^1$, with equal probability $p(x,y) = \frac 14$ from Eq. (\ref{Eq:probxyzs}), which he wishes to share with Alice securely. After the proper execution of the protocol Alice's key string $K_{A_1} \in \{i,j\}_{i,j=0}^1$ is highly correlated with $K_{B_1}$, and it reaches a maximum $K_{A_1} = K_{B_1}$ in the ideal situation .
Moreover, in order to estimate the second term on the r.h.s of (\ref{eq:A1B1keyrate}),  we assume that Bob2 is an honest party who follows the protocol. In this process, Bob2 will announce publicly his key bit $K_{B_2} \in \{z\}_{z = 0}^1$ which he encodes in his part of the shared $\rho_{AB_1B_2}$, with equal probability. Hence from  (\ref{eq:A1B1keyrate}), we have
 

\begin{eqnarray}
 r({\cal P}_1) &\geq& 
 \frac 12 \sum_s \big( I(K_{A_1}:K_{B_1})_{\kappa^s} - I(K_{A_1}:E|K_{B_2})_{\kappa^s} - \delta_1(s) - 2\delta_2(s) \big), \nonumber \\
 &=& \frac 12 \sum_s \big( H(K_{A_1})_{\kappa^s} - H(K_{A_1}|K_{B_1})_{\kappa^s} - H(K_{A_1}|K_{B_2})_{\kappa^s} + H(K_{A_1}|EK_{B_2})_{\kappa^s} -\delta_1(s) - 2\delta_2(s) \big), \label{eq:local_A1B1_1} \nonumber \\
 &\geq&  \frac 12 \sum_s \big( H(K_{A_1}|EK_{B_2})_{\kappa^s} - H(K_{A_1}|K_{B_1})_{\kappa^s}-\delta_1(s) - 2\delta_2(s) \big),  \label{eq:local_A1B1_2} \nonumber \\
 &=& \frac 12 \sum_s \left( \frac 12 \sum_z H(K_{A_1}|E)_{\kappa^s(z)} - H(K_{A_1}|K_{B_1})_{\kappa^s}-\delta_1(s) - 2\delta_2(s) \right).  \label{eq:local_A1B1_3}
 \end{eqnarray}
 In the second inequality we have used the fact that $H(K_{A_1})_{\kappa^s} - H(K_{A_1}|K_{B_2})_{\kappa^s} = I(K_{A_1}:K_{B_2})_{\kappa^s} \geq 0$, and by $\kappa^s$, in the subscript of the above expressions, we refer to the ccq state, $\kappa^s_{K_{A_1}K_{A_2}K_{B_1}K_{B_2}EE'}$, given in Eq. (\ref{Eq:kappa_def3}),
 and by $\kappa^s(z)$,  we denote 
 \begin{eqnarray}
   \kappa^s_{K_{A_1}K_{A_2}K_{B_1}EE'}(z) = \sum_{ijk} \sum_{xy}  p(ijk;xy|zs) \ket{ij}\bra{ij}_{K_{A_1}} \otimes  \ket{xy}\bra{xy}_{K_{B_1}} \otimes \ket{k}\bra{k}_{K_{A_2}}  \otimes \rho_{EE'}^{ijk;xy;zs},
 \end{eqnarray}
 where the probability of getting outcome $z$, $p(z) = \frac 12$, $p(ijk;xy|zs) = 4 p(ijk;xy;zs)$. According to our protocol, the key bits between Bob1 and Bob2 are completely uncorrelated, hence, $\delta_3 = 0$.
 
 Computing $H(K_{A_1}|K_{B_1})_{\kappa^s}$, and $\delta_i(s)$ for $i = 1,2$ is not very difficult once we know $\kappa^s_{K_{A_1}K_{A_2}K_{B_1}K_{B_2}EE'}$, as $\delta_i(s) \geq H(K_{A_i}|K_{B_i})_{\kappa^s}, ~i = 1,2$ from (\ref{eq4:as1}) and (\ref{eq4:as2}). The lower bound in (\ref{eq:A1B1keyrate}), holds true for all $\delta_i(s) \geq H(K_{A_i}|K_{B_i})_{\kappa^s}, ~i = 1,2$, and therefore
 \begin{eqnarray}
    r({\cal P}_1) &\geq& \frac 14   \sum_{z,s} H(K_{A_1}|E)_{\kappa^s(z)} - \frac 12 \sum_s  (H(K_{A_1}|K_{B_1})_{\kappa^s} + H(K_{A_2}|K_{B_2})_{\kappa^s} ).
 \end{eqnarray}
 
 The non-trivial term in the above expression is $H(K_{A_1}|E)_{\kappa^s(z)}$, and
 note that in order to compute it we no longer need to consider  $K_{A_2}$. As such,  the minimal required ccq state is 
 \begin{eqnarray}\label{eq:kappa_A1B1E}
   && \kappa^s_{K_{A_1}K_{B_1}EE'}(z) \nonumber \\
   &=& \sum_{ij} \sum_{xy}  \sum_k p(ijk;xy|zs) \ket{ij}\bra{ij}_{K_{A_1}} \otimes  \ket{xy}\bra{xy}_{K_{B_1}}  \otimes \rho_{EE'}^{ijk;xy;zs}, \nonumber \\
   &=&\sum_{ij} \sum_{xy}  \sum_k  tr_{AA_1A_2B_1B'_1} \left( (\hat{K}^{i,j,k}_s)_{AA_1A_2} \otimes \hat{K}^{x,y}_{B_1B_1'} \otimes \rho^{z,s}_{AA_1A_2B_1B'_1EE'} (\hat{K}^{i,j,k}_s)_{AA_1A_2} \otimes \hat{K}^{x,y}_{B_1B_1'} \right) \ket{ij}\bra{ij}_{K_{A_1}} \otimes  \ket{xy}\bra{xy}_{K_{B_1}} \nonumber \\
   &=& \sum_{ij} \sum_{xy}   tr_{AA_1A_2B_1B'_1} \left( (\hat{\bar{K}}^{i,j}_s)_{AA_1A_2} \otimes \hat{K}^{x,y}_{B_1B_1'} \otimes \rho^{z,s}_{AA_1A_2B_1B'_1EE'} (\hat{\bar{K}}^{i,j}_s)_{AA_1A_2} \otimes \hat{K}^{x,y}_{B_1B_1'} \right) \ket{ij}\bra{ij}_{K_{A_1}} \otimes  \ket{xy}\bra{xy}_{K_{B_1}} 
 \end{eqnarray}
 where we have used Eq. (\ref{eq:varrho_def}) for $p(ijk;xy|zs)$, $\rho^{z,s}_{AA_1A_2B_1B'_1EE'} = 4 \tr_{B_2B'_2}(\ket{B(z,s)}\bra{B(z,s)}_{B_2B'_2}\ket{\Psi}\bra{\Psi}_{AA_1A_2B_1B'_1B_2B'_2EE'})$ and $\hat{\bar{K}}^{i,j}_s = \sum_k \hat{K}^{i,j,k}_s$, for $s = 0,1$, is rank-2 projective measurement.
 
 Let us now consider the test run of the protocol, which both the honest parties perform to estimate the correlation present between the honest parties and the eavesdropper. Since each of the honest parties perform the test run with probability $1 - p$, then the ccq state when only Bob1 perform the test run is as follows
 \begin{eqnarray}\label{Eq:xi_defAB1}
  && \xi^s_{K_{A_1}T_{B_1}EE'}(z) \nonumber \\
   &=& \sum_{ij} \sum_{\bar{x}\bar{y}}   tr_{AA_1A_2B_1B_2} \left( (\hat{\bar{K}}^{i,j}_s)_{AA_1A_2} \otimes \hat{T}^{x,y}_{B_1B_1'} \otimes \rho^{z,s}_{AA_1A_2B_1B'_1EE'} (\hat{\bar{K}}^{i,j}_s)_{AA_1A_2} \otimes \hat{T}^{x,y}_{B_1B_1'} \right) \ket{ij}\bra{ij}_{K_{A_1}} \otimes  \ket{x,y}\bra{x,y}_{T_{B_1}}, \nonumber \\
 \end{eqnarray}
 where, $\hat{T}^{x,y}_{B_1B_1'} = \ket{x,y_\vdash}\bra{x,y_\vdash}_{B_1B'_1}$ are the projective measurements for the test run and  $T_{B_1}$ represents the system of the register, which stores the measurement outcome of the test run of Bob1.
 
 
From the no-singalling principle we find that $H(K_{A_1}|E)_{\kappa^s(z)} =
H(K_{A_1}|E)_{\xi^s(z)}$.
Moreover,  when Alice also perform test run the ccq state takes the form 
\begin{eqnarray}\label{Eq:sigma_defAB1}
  && \sigma^s_{T_{A_1}T_{B_1}EE'}(z) \nonumber \\
   &=& \sum_{i,j} \sum_{x,y}   tr_{AA_1A_2B_1B_2} \left( (\hat{T}^{i,j})_{AA_1A_2} \otimes \hat{T}^{x,y}_{B_1B_1'} \otimes \rho^{z,s}_{AA_1A_2B_1B'_1EE'} (\hat{T}^{i,j})_{AA_1A_2} \otimes \hat{T}^{x,y}_{B_1B_1'} \right) \ket{i,j}\bra{i,j}_{T_{A_1}} \otimes  \ket{x,y}\bra{x,y}_{T_{B_1}}, \nonumber \\
 \end{eqnarray}
 The rank two projective measurements performed by Alice on her entire 3-qubit system are given by  $\{ (\hat{T}^{i,j})_{AA_1A_2}\}_{i,j = 0}^1$,
 \begin{equation}\label{eq:rank2_test}
     (\hat{T}^{i,j})_{AA_1A_2} =  \ket{i}\bra{i}_A \otimes \ket{j_\vdash}\bra{j_\vdash}_{A_1} \otimes \mathbb{I}_{A_2},
 \end{equation}
 where $\ket{\beta_\vdash} = \frac{1}{\sqrt{2}}(\ket{0} + (-1)^\beta \ket{1})$.
 
 Now applying the entropic uncertainty relation \cite{Berta2010,tomamichel2012framework}, in the two states $\sigma^s_{T_{A_1}T_{B_1}EE'}(z)$ and $\xi^s_{K_{A_1}T_{B_1}EE'}(z)$, where  two different sets of measurements $\{ (\hat{\bar{K}}^{i,j}_s)_{AA_1A_2}\}_{i,j = 0}^1$ and $\{ (\hat{T}^{i,j})_{AA_1A_2}\}_{i,j = 0}^1$ are applied to the subsystem kept by Alice gives
 \begin{equation}
     H(K_{A_1}|E)_{{\xi^s}(z)} + H(T_{A_1}|T_{B_1})_{{\sigma^s}(z)} \geq \log_2 \frac 1c,
 \end{equation}
 where, $c$ is the overlap between the two set of measurement operators. It is given by \cite{tomamichel2012framework}
 \begin{eqnarray}
     c = \max_{(i,j), (i'j')} ||\sqrt{(\hat{K}_s^{i,j})_{AA_1A_2}}\sqrt{(\hat{T}^{i',j'})_{AA_1A_2}}||_{\infty}^2 = \frac 14. 
     \label{eq:overlap1}
 \end{eqnarray}
 The proof of Eq. (\ref{eq:overlap1}), has been given in Appendix \ref{appen:overlap}.
 Finally,  the lower bound on the local key rate between  Alice to Bob1 can be expressed as 
 
 \begin{eqnarray}
 r({\cal P}_1)  &\geq& 2 - \frac 14   \sum_{z,s} H(T_{A_1}|T_{B_1})_{{\sigma^s}(z)} - \frac 12 \sum_s  (H(K_{A_1}|K_{B_1})_{\kappa^s} + H(K_{A_2}|K_{B_2})_{\kappa^s} )\label{eq:final_rate_A1B1}.
 \end{eqnarray}
 
We calculate the above bound for different noise models in Sec. \ref{sec:noisy_case}.
 

 
 \subsection{Local key rates: Alice to Bob2}\label{sec:ratesA2B2}
 In our protocol, Bob2 prepares a single bit $K_{B_2} \in \{z\}_{z = 0}^1$, which he wants to share with Alice, by performing a unitary encoding  - or equivalently a measurement in the purified protocol - on the shared state. Alice performs a measurement on the encoded state and obtains information about $K_{A_2} \in \{k\}_{k = 0}^1$. In the ideal execution of the protocol we always have $K_{A_2} = K_{B_2}$, but due to the presence of the eavesdropper, the amount of secure key they can share between each other is less than this. It is given in Eq. (\ref{eq:A2B2keyrate}), which can be written as 
 \begin{eqnarray}
 r({\cal P}_2) &\geq&
\frac 12 \sum_s \big( I(K_{A_2}:K_{B_2})_{\kappa^s} - I(K_{A_2}:E|K_{B_1})_{\kappa^s} - 2\delta_1(s) - \delta_2(s) \big) \label{eq:local_A2B_1}, \nonumber \\
&=& \frac 12 \sum_s (H(K_{A_2})_{\kappa^s} - H(K_{A_2}|K_{B_2})_{\kappa^s} - H(K_{A_2}|K_{B_1})_{\kappa^s} + H(K_{A_2}|EK_{B_1})_{\kappa^s} - 2\delta_1(s) - \delta_2(s)), \nonumber \\
&\geq&  \frac 12 \sum_s \big( H(K_{A_2}|EK_{B_1})_{\kappa^s} - H(K_{A_2}|K_{B_2})_{\kappa^s} - 2\delta_1(s) - \delta_2(s) \big),  \label{eq:local_A2B2_2} \\
 &=& \frac 12 \sum_s \left( \frac 14 \sum_{xy} H(K_{A_2}|E)_{\kappa^s(xy)} - H(K_{A_2}|K_{B_2})_{\kappa^s}-2\delta_1(s) - \delta_2(s) \right),  \label{eq:local_A2B2_3}
 \end{eqnarray}
 where $\kappa^s$ in the subscript denotes  $\kappa^s_{K_{A_1}K_{A_2}K_{B_1}K_{B_2}EE'}$, and we use $H(K_{A_2})_{\kappa^s} - H(K_{A_2}|K_{B_1})_{\kappa^s} = I(K_{A_2}:K_{B_1}) \geq 0$ in the second inequality (\ref{eq:local_A2B2_2}).
 The ccq state $\kappa^s(x,y)$, in the subscript of the first term in (\ref{eq:local_A2B2_3})  is
 \begin{eqnarray}\label{eq:kappa_A1A2B2E}
   \kappa^s_{K_{A_1}K_{A_2}K_{B_2}EE'}(x,y) = \sum_{ijk}  \sum_{z} p(ijk;z|xys) \ket{ij}\bra{ij}_{K_{A_1}} \otimes \ket{k}\bra{k}_{K_{A_2}} \otimes \ket{z}\bra{z}_{K_{B_2}} \otimes \rho_{EE'}^{ijk;xy;zs},
 \end{eqnarray}
 and the factor $\frac 14 = p(x,y)$ represents the probability of getting outcome $\{x,y\}$ in the key generation run of Bob1.
 Note that in order to find the lower bound on $r({\cal P}_2)$, one needs to estimate the correlation between Alice and Eve (the first term in Eq. (\ref{eq:local_A2B2_3})), when the value of $K_{B_1} \in \{x,y\}_{x,y = 0}^1$, the raw key string  of Bob1, is known. This means that Bob1 needs to disclose his bits publicly. This was also the case when we were estimating the value of $H(K_{A_1}|EK_{B_2})_{\kappa^s}$, the lower bound on $r({\cal P}_1)$. 
 
 
 Let us now focus  on the second term in Eq. (\ref{eq:local_A2B2_3}) i.e. $H(K_{A_2}|E)_{\kappa^s(x,y)}$. 
 The effective ccq state is 
 \begin{eqnarray}\label{eq:kappa_A2B2E}
   \kappa^s_{K_{A_2}K_{B_2}EE'}(x,y) = \sum_{k}  \sum_{z} \left(\sum_{ij} p(ijk;z|xys)  \rho_{EE'}^{ijk;xy;zs} \right) \otimes \ket{k}\bra{k}_{K_{A_2}} \otimes \ket{z}\bra{z}_{K_{B_2}},
 \end{eqnarray}
 where $p(ijk;z|xys) = 4 p(ijk;xy;z|s)$, and from Eqs. (\ref{Eq:kappa_def}) and (\ref{eq:varrho_def}), we can write $\kappa^s_{K_{A_2}K_{B_2}EE'}(x,y)$, as
 \begin{eqnarray}\label{eq:kappa_A2B2E_2}
   && \kappa_{K_{A_2}K_{B_2}EE'}^s(x,y) \nonumber \\
   &=& \sum_{k}  \sum_{z} \sum_{ij} 
   \tr_{AA_1A_2B_2B'_2} \left( (\hat{K}^{i,j,k}_s)_{AA_1A_2} \otimes \hat{K}^{z,s}_{B_2B_2'} ~ \rho^{x,y}_{AA_1A_2B_2B'_2EE'} ~(\hat{K}^{i,j,k}_s)_{AA_1A_2} \otimes \hat{K}^{z,s}_{B_2B_2'} \right) \otimes \ket{k}\bra{k}_{K_{A_2}} 
   \otimes \ket{z}\bra{z}_{K_{B_2}}, \nonumber \\
    &=& \sum_{k}  \sum_{z} \tr_{AA_1A_2B_2B'_2} \left( (\hat{\tilde{K}}^{k}_s)_{AA_1A_2} \otimes \hat{K}^{z,s}_{B_2B_2'}  ~\rho^{x,y}_{AA_1A_2B_2B'_2EE'}~ (\hat{\tilde{K}}^{k}_s)_{AA_1A_2} \otimes \hat{K}^{z,s}_{B_2B_2'} \right) \otimes \ket{k}\bra{k}_{K_{A_2}} 
   \otimes \ket{z}\bra{z}_{K_{B_2}}, 
 \end{eqnarray}
 where $\rho^{x,y}_{AA_1A_2B_2B'_2EE'} =  4 \tr_{B_1B'_1}(\hat{K}^{x,y}_{B_1B'_1}\ket{\Psi}\bra{\Psi}_{AA_1A_2B_1B'_1B_2B'_2EE'}\hat{K}^{x,y}_{B_1B'_1})$ and $\hat{\tilde{K}}^{k}_s = \sum_{i,j} \hat{K}^{i,j,k}_s$, is rank-4 projective measurement.

 Now consider the following state
 \begin{eqnarray}\label{eq:xi_A2B2E_2}
   && \xi^{s}_{K_{A_2}T_{B_2}EE'}(x,y) \nonumber \\& =& \sum_{k}  \sum_{z} \tr_{AA_1A_2B_2B'_2} \left( (\hat{\tilde{K}}^{k}_s)_{AA_1A_2} \otimes \hat{T}^{z,s}_{B_2B_2'} \rho^{x,y}_{AA_1A_2B_2B'_2EE'} (\hat{\tilde{K}}^{k}_s)_{AA_1A_2} \otimes \hat{T}^{z,s}_{B_2B_2'} \right) \ket{k}\bra{k}_{K_{A_2}} 
   \otimes \ket{z}\bra{z}_{T_{B_2}} ~~~~~~~ 
 \end{eqnarray}
 where for each $s$ the projective measurement $\hat{T}^{z,s}_{B_2B_2'} = \ket{z_\vdash, (z \oplus s)_\vdash}\bra{z_\vdash, (z \oplus s)_\vdash}_{B_2B_2'}$ and $\ket{\beta_\vdash} = \frac{1}{\sqrt{2}}(\ket{0} + (-1)^\beta \ket{1})$. It is easy to verify that for each value of $s \in \{0,1\}$, the two different measurements lie in the same subspace of the Hilbert space $({\cal H}^2)^{\otimes 2}$, i.e.  $\sum_z \hat{K}^{z,s}_{B_2B_2'}  = \sum_z \hat{T}^{z,s}_{B_2B_2'}$. It implies that $\kappa_{A_2EE'}^s(x,y) = \xi_{A_2EE'}^s(x,y)$ and in return $H(A_2|E)_{\kappa^s(x,y)} = H(A_2|E)_{\xi^s(x,y)}$.
 
The state of the system is given by ccq state from  Eq. (\ref{eq:xi_A2B2E_2}) when Bob2 applies the test run. This is irrespective of the key generation run performed by the other parties. Recall that Bob2's measurement basis in the test run is different from Bob1, as Bob2 shares only a single key bit with Alice, and he discloses his auxiliary bit $s$. It divides the 4-dimensional space in part of Bob2 into two sub-spaces. Moreover $\{K^{z,s}\}_{z=0}^1$ and $\{T^{z,s}\}_{z=0}^1$ form two complementary measurements basis for both $s =0,1 $.

 Alice performs her test run to detect an eavesdropper in the channel between her and Bob2 by applying measurements to her joint system. The respective projective measurement $(\hat{\tilde{T}}^{k}_s)_{AA_1A_2}$ is given by 
 \begin{equation}
      (\hat{\tilde{T}}^{k}_s)_{AA_1A_2} = \mathbb{I}_{A}\otimes \mathbb{I}_{A_1}\otimes \ket{(k \oplus s)_\vdash}\bra{(k \oplus s)_\vdash}_{A_2}.
  \end{equation}
It results in the following ccq state
 \begin{eqnarray}\label{eq:sigma_A2B2E_2}
   && \sigma^{s}_{T_{A_2}T_{B_2}EE'}(x,y) \nonumber \\& =& \sum_{k}  \sum_{z} \tr_{AA_1A_2B_2B'_2} \left( (\hat{\tilde{T}}^{k}_s)_{AA_1A_2} \otimes \hat{T}^{z,s}_{B_2B_2'} \rho^{x,y}_{AA_1A_2B_2B'_2EE'} (\hat{\tilde{T}}^{k}_s)_{AA_1A_2} \otimes \hat{T}^{z,s}_{B_2B_2'} \right) \ket{k}\bra{k}_{T_{A_2}} 
   \otimes \ket{z}\bra{z}_{T_{B_2}}, 
 \end{eqnarray}
 where $T_{A_2}$ is a register in which Alice keeps her measurement outcome of the test run.  
 
 Applying the entropic uncertainty relation to states $\xi^{s}_{K_{A_2}T_{B_2}EE'}(x,y)$ and $\sigma^{s}_{T_{A_2}T_{B_2}EE'}(x,y)$, we find 
 \begin{equation}\label{eq:entropic_uncertainty2}
     H(K_{A_2}|E)_{\xi^{s}(x,y)} + H(T_{A_2}|T_{B_2})_{\sigma^{s}(x,y)} \geq \log_2 \frac{1}{\tilde{c}},
 \end{equation}
 Quantity $\tilde{c}$ is the overlap between two measurements $\{(\hat{\tilde{K}}^{k}_s)_{AA_1A_2}\}$ and $\{(\hat{\tilde{T}}^k_s)_{AA_1A_2}\}$ by Alice, is given by 
 \begin{eqnarray}
     \tilde{c} = \max_{k', k} ||\sqrt{(\hat{\tilde{K}}_s^{k'})_{AA_1A_2}}\sqrt{(\hat{\tilde{T}}^k_s)_{AA_1A_2}}||_{\infty}^2 = \frac 12. \label{eq:overlap2}
 \end{eqnarray}
 See Appendix \ref{appen:overlap} for the proof.
 Putting this altogether we obtain the final expression for the key rate 
 \begin{eqnarray}
     r({\cal P}_2) &\geq& 1 - 
  \frac 12 \sum_s \left( \frac 14 \sum_{xy} H(T_{A_2}|T_{B_2})_{\sigma^s(xy)} - H(K_{A_2}|K_{B_2})_{\kappa^s}-2\delta_1(s) - \delta_2(s) \right), \label{eq:semifinal_rate_A2B2}\\
  &=& 1 -  \frac 18 \sum_{xys} H(T_{A_2}|T_{B_2})_{\sigma^s(xy)} - \frac 12 \sum_s (H(K_{A_1}|K_{B_1})_{\kappa^s} + H(K_{A_2}|K_{B_2})_{\kappa^s}), \label{eq:final_rate_A2B2}
 \end{eqnarray}
 where we have replaced $\delta_1(s) = H(K_{A_1}|K_{B_1})_{\kappa^s}$ and $\delta_2(s) = H(K_{A_2}|K_{B_2})_{\kappa^s}$, as (\ref{eq:semifinal_rate_A2B2}) is true for all $\delta_1(s) \geq  H(K_{A_1}|K_{B_1})_{\kappa^s}$ and $\delta_2(s) \geq H(K_{A_2}|K_{B_2})_{\kappa^s}$ (from (\ref{eq4:as1}) and (\ref{eq4:as2})).

 \section{Key rates for noisy quantum channels} 
\label{sec:noisy_case}

In this section, we are going to analyse the different strategies of an eavesdropper. We will calculate the lower bound on the key rate for several exemplary quantum channels. According to our protocol, Alice and Bobs use forward and backward quantum channels to transfer two qubits where each sender receives only a single qubit.

We consider two main common noise models, namely the depolarising channel and the amplitude damping channel, under different scenarios and effects 
on forward and backward communication and calculate the respective lower bounds on the key rate.

 \subsection{Depolarising channel}
Quantum system in a state $\rho$ when passing through the depolarising channel, is transformed into ${\cal D}^\lambda(\rho) = (1 - \lambda) \rho + \frac{\lambda}{d}\mathbb{I}_d$, where $d$ is the dimension of the density matrix and $\lambda \in (0,1)$, is the noise parameter.
In the case of a qubit system, $\rho \in {\cal H}^2$, the depolarising channel can be written down in terms of the Kraus operators as
\begin{equation}
    \label{eq:depol_def}
    {\cal D}^\lambda(\rho) = (1 - 3\frac{\lambda}{4}) \rho +  \frac{\lambda}{4}(\sigma_x \rho \sigma_x +\sigma_y \rho \sigma_y + \sigma_z \rho \sigma_z),
\end{equation}
where $\lambda \in (0,1)$ is the noise parameter. 
We consider several scenarios of how the noise can affect the system. Remember that the effect of noise in the transmission channel  ${\cal E}^f_{\substack{A_1 \rightarrow X_1 \\ A_2 \rightarrow X_2}}$ and ${\cal E}^b_{\substack{X_1 \rightarrow A_1 \\ X_2 \rightarrow A_2}}$ can be considered as the effect of a possible eavesdropper. If Eve is close to Alice, she can apply the correlated noise model to the forward and backward channel. In contrast, independent noise is much more realistic if two different eavesdropper acts separately on the channel.

\subsubsection{Independent depolarising channel}
Let us first consider a scenario in which channels between Alice and two Bobs are attacked by two separate eavesdroppers. We can model the forward transmission channel in terms of two independent depolarising channels as 
\begin{equation}\label{eq:forward_depol}
    {\cal E}^f_{\substack{A_1\rightarrow B_1 \\ A_2\rightarrow B_2 }} = {\cal D}^{\lambda}_{A_1 \rightarrow B_1} \otimes {\cal D}^{\delta}_{A_2 \rightarrow B_2}.
\end{equation}
Moreover, let us assume that the same noise acts when the senders send their encoded part of the shared state back to Alice, so that
\begin{equation}\label{eq:backward_depol}
    {\cal E}^b_{\substack{B_1\rightarrow A_1 \\ B_2\rightarrow A_2 }} = {\cal D}^{\lambda}_{B_1 \rightarrow A_1} \otimes {\cal D}^{\delta}_{B_2 \rightarrow A_2}.
\end{equation}
 
 The conditional probability distribution $p(ijk|xy;zs)$, after all the honest parties perform the key generation run, is independent of   
 the auxiliary bit $s \in \{0,1\}$. Moreover, the conditional probability $p(ijk|xy;zs)$ takes the form $ \mathbb{P}(i,j,k)=p(i \oplus x, j\oplus y, k \oplus z|xy;zs)$, where
 \begin{eqnarray}
\mathbb{P}(0,0,0) &=& 1+ \frac{5}{8} \lambda \delta (2 - \lambda) (2 - \delta) - \frac{3}{4} (\lambda(2 - \lambda) + \delta(2 - \delta)), \\
\mathbb{P}(0,0,1) &=& \frac{1}{8}   \left(2 - \lambda( 2 -  \lambda) \right) \delta(2-\delta),\\
\mathbb{P}(0,1,0) &=& \frac{1}{4} (\lambda(2 - \lambda) + \delta(2 - \delta)) - \frac{3}{8} \lambda \delta (2 - \lambda) (2 - \delta),\\
\mathbb{P}(0,1,1) &=& \frac{1}{8}  \left(2 - \lambda( 2 -  \lambda) \right)\delta (2 - \delta), \\
\mathbb{P}(1,0,0) &=& \frac{1}{8} \left( 2 - \delta(2 - \delta)\right) \lambda (2 -\lambda),\\
\mathbb{P}(1,0,1) &=& \frac{1}{8} \lambda\delta (2-\lambda)  (2 -\delta),  \\
\mathbb{P}(1,1,0) &=& \frac{1}{8} \left(2 - \delta(2 - \delta) \right) \lambda (2-\lambda),\\
\mathbb{P}(1,1,1) &=& \frac{1}{8} \lambda\delta (2-\lambda)  (2 -\delta).
\end{eqnarray}
One can easily calculate $H(K_{A_1}|K_{B_1})_{\kappa^s}$ and $H(K_{A_2}|K_{B_2})_{\kappa^s}$ for the classical-classical state shared between the honest parties in the presence of the independent depolarising noise. To calculate the lower bound on $r({\cal P}_1)$, we need $H(T_{A_1}|T_{B_1})_{{\sigma^s}(z)}$. ${\sigma^s}(z)_{T_{A_1}T_{B_1}}$ describes the classical-classical state shared among Alice, Bob1 and Bob2,
when Alice and Bob1 perform the test run and Bob2 performs key generation run, and disclose the values of  $K_{B_2}$, and $s$, 
\begin{eqnarray}
   \sigma^s_{T_{A_1}T_{B_1}}(z) 
   = \sum_{i,j} \sum_{x,y}   q(ij|xy;zs) \ket{i,j}\bra{i,j}_{T_{A_1}} \otimes  \ket{x,y}\bra{x,y}_{T_{B_1}},
\end{eqnarray}
where $q(ij|xy;zs) =q(i,j;x,y;z,s)/ q(x,y;z,s) $, with
$$q(i,j;x,y;z,s) = \tr \left( (\hat{T}^{i,j})_{AA_1A_2} \otimes \hat{T}^{x,y}_{B_1B_1'} \otimes \hat{K}^{z,s}_{B_2B_2'} \otimes \ket{\Psi}\bra{\Psi}_{AA_1A_2B_1B'_1EE'} (\hat{T}^{i,j})_{AA_1A_2} \otimes \hat{T}^{x,y}_{B_1B_1'} \otimes \hat{K}^{z,s}_{B_2B_2'} \right).$$
For independent depolarising noise in both forward and backward transmission channel, we find the probabilities $q(x,y;z,s) = \sum_{ij} q(i,j;x,y;z,s) = \frac{1}{16}$. Moreover $q(i,j;x,y;z,s)$ is independent of $z$ and $s$.. Furthermore $q(i\oplus x, j \oplus y|xy;zs) = \mathbb{Q}(i,j)$, where 
\begin{eqnarray}
   \mathbb{Q}(0,0) &=& \frac{1}{4} (2 - \lambda)^2, \\
   \mathbb{Q}(0,1) &=& \frac{1}{4} \lambda (2 - \lambda ), \\
   \mathbb{Q}(1,0) &=& \frac{1}{4} \lambda (2 - \lambda ), \\
   \mathbb{Q}(1,1) &=& \frac{\lambda ^2}{4}. \\
\end{eqnarray}
We have numerically calculated the lower bounds on the key rate  $r({\cal P}_1)$ between Alice and Bob1 with the help of $\mathbb{P}$ and $\mathbb{Q}$, and ploted it in the plane of the noise parameter $\delta$ and $\lambda$, which is given in the left panel in figure \ref{fig:ind_depol}. Note that even though $H(T_{A_1}|T_{B_1})_{{\sigma^s}(z)}$ is independent of $\delta$, the conditional entropies $H(K_{A_1}|K_{B_1})_{\kappa^s}$ and $H(K_{A_2}|K_{B_2})_{\kappa^s}$ are not. 
 \begin{figure}[t]
\centering
 \includegraphics[width=1\columnwidth,keepaspectratio,angle=0]{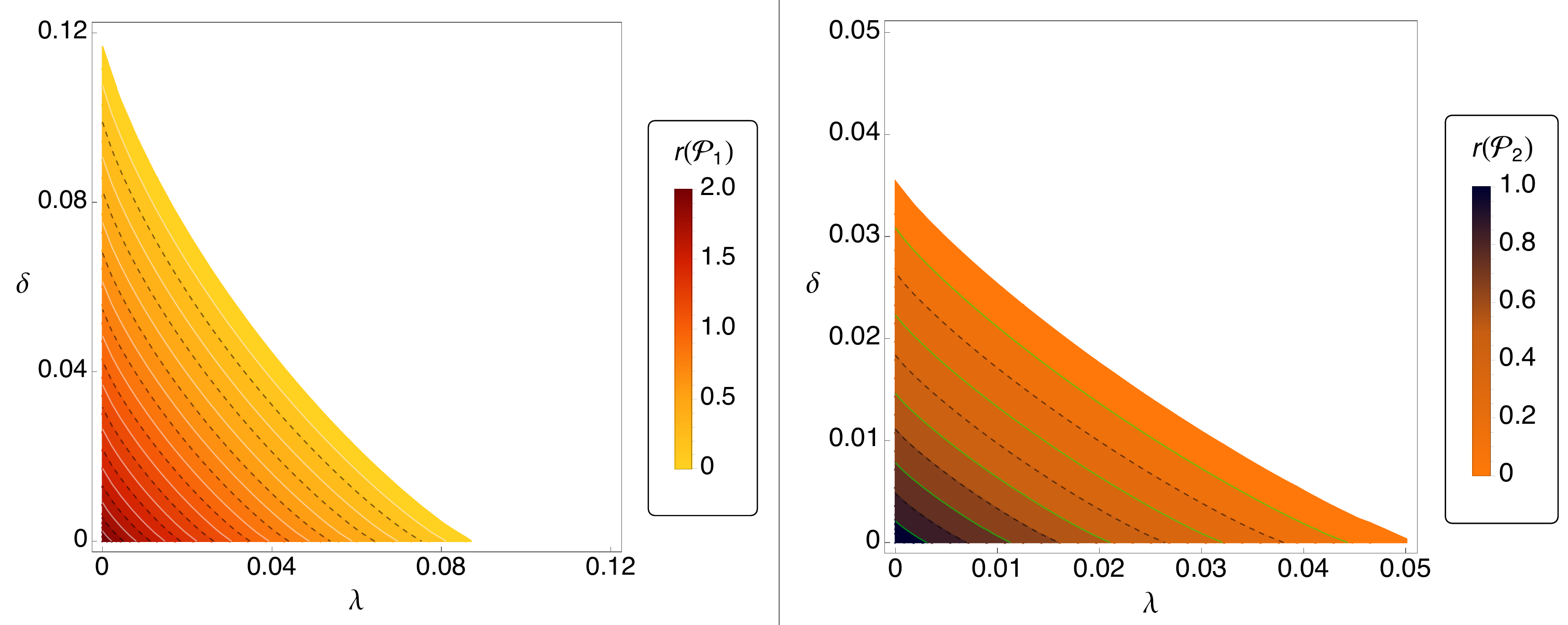}
\caption{Plots of the lower bounds on the local key rates $r({\cal P}_1)$ (the left panel) and $r({\cal P}_2)$ (the right panel), when the independent depolarising noise is affecting the forward and backward transmission quantum channels.
We have used the same noise model for both the forward and backward channel. 
Here, $\lambda$ and $\delta$ are the parameters of the independent depolarising channel as given by Eqs. (\ref{eq:forward_depol}) and (\ref{eq:backward_depol}).}
\label{fig:ind_depol}
 \end{figure}

We have also calculated numerically the lower bound on the local key rates between Alice to Bob2, $r({\cal P}_2)$ in the right panel of figure \ref{fig:ind_depol}.
In order to calculate this, we  additionally need  $H(T_{A_2}|T_{B_2})_{\sigma^s(x,y)}$, where $\sigma^s_{T_{A_2}T_{B_2}}(x,y)$ is the classical-classical state after Alice and Bob2 perform their measurements for the test run and Bob1 discloses his bits $K_{A_1}$ for key generation run. Given $s$, it can be expressed as 
 \begin{eqnarray}
   \sigma^s_{T_{A_2}T_{B_2}}(x,y)
   = \sum_{k} \sum_{z}   \tilde{q}(k|xy;zs) \ket{k}\bra{k}_{T_{A_2}} \otimes  \ket{z}\bra{z}_{T_{B_2}},
\end{eqnarray}
where as before, $\tilde{q}(x,y;z,s) = \frac{1}{16}$ is the probability of Bob2's test measurement outcome and Bob1 key measurement while $\tilde{q}(k|xy;zs):= \tilde{q}(k \oplus z|xy;zs):=  \left\{\begin{array}{l}
        1 - \frac{\delta}{2}, ~~k = 0 \\
        \frac{\delta}{2}, ~~k = 1
        \end{array}\right. . $
        
 We have observed that $r({\cal P}_1)$ is susceptible to more noise than $r({\cal P}_2)$. Moreover, the lower bound on the key rate of $r({\cal P}_1)$ is affected more by the noise in the channel connecting Alice and the Bob1 than by the noise connecting Alice with Bob2. And the same is true for $r({\cal P}_2)$ also.

 In the noiseless scenario, when both the $\lambda = \delta = 0$, the bound on $r({\cal P}_1)$ reaches $2$, and $r({\cal P}_2)$ becomes $1$. Hence, in total, both the senders can share $3$ bits of secure key with the receiver, which is also the dense coding capacity of the pure GHZ state \cite{GHZk, BrussALMSS2005-multidense}. This result shows that our protocol of securing multiparty dense coding protocol is consistent with the capacity of classical information transmission without security. 
 
 Note that for in the dense coding capacity of an arbitrary state, one optimizes  over all possible encoding and decoding \cite{Hiroshima_2001,HoroCapacity,BrussALMSS2005-multidense,ShadmanNoise, DDCReznik,ROYDDC}. Whereas in our consideration, we stick to an encoding operation (the Pauli matrices) for both Bobs and single decoding of Alice (the GHZ basis measurements) independent of the noise present in the system.
 Moreover, it would be interesting to see how the amount of classical information one can transfer with security (key rate) compares to the total dense coding capacity in the presence of noise. However, this analysis is beyond the scope of our current manuscript.
 
 \subsubsection{Independent depolarising channel acting only one side}

We now analyse a situation when the independent depolarising channel is acting either at the time of transmission of $\ket{GHZ}_{AA_1A_2}$ or on the encoded pure state. If it acts  only on the forward channel, then ${\cal E}^f_{\substack{A_1\rightarrow B_1 \\ A_2\rightarrow B_2 }}$ is  given by Eq. (\ref{eq:forward_depol}) and ${\cal E}^b_{\substack{B_1\rightarrow A_1 \\ B_2\rightarrow A_2 }}$ is idempotent. In reverse situation, the forward channel is idempotent and Eq. (\ref{eq:backward_depol}) provides the expression for the backward channel.

Firstly, it is apparent that the bounds on $r({\cal P}_1)$ decrease more slowly than compared to the situation when both forward and backward channels are subject to noise. This is because we restrict the total effective noise acting on the system. Secondly, the bounds provide the same numerical value when we apply noise in the forward or only in the backward channel. This phenomenon is quite obvious because the depolarising channel is a covariance noise, i.e., it commutes with the encoded unitary operators. Hence, both $H(K_{A_i}|K_{B_i})_{\kappa^s}$ for  $i = 1,2$ are the same for both noise models. Moreover, according to protocol, the measurement operators of Alice and Bob1 are chosen such that   $H(T_{A_1}|T_{B_1})_{\sigma^s(z)}$ does not change with $z$ and $s$. 

However, the above equality does not hold for the lower bound on $r({\cal P}_2)$ due to the choice of the projective measurement for the test run. We have observed that for a fixed $\lambda$ and $\delta$,  the bound 
on $r({\cal P}_2)$ is lower for the noisy backward channel than for the noisy forward channel.



\begin{figure}[t]
\centering
 \includegraphics[width=1\columnwidth,keepaspectratio,angle=0]{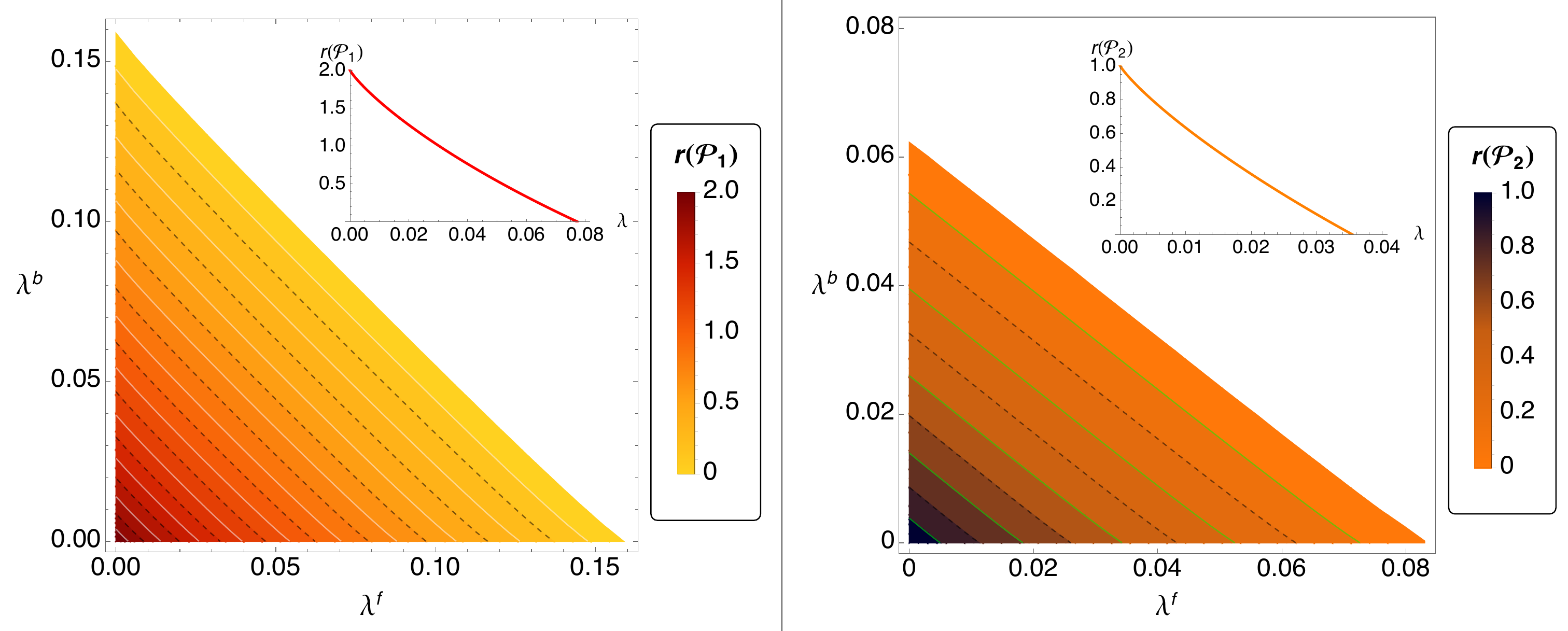}
\caption{Plot of the lower bounds on the local key rates $r({\cal P}_1)$  and $r({\cal P}_2)$, when fully correlated depolarising noise is affecting the forward and backward transmission quantum channels.
Here we have used a completely independent noise model for both the forward and backward channels, which are parametrized by $\lambda^f$ and $\lambda^b$. 
In the left panel, we have plotted the lower bound of $r({\cal P}_1)$ in the plane of  $\lambda^f$ and $\lambda^b$, whereas in the inset figure, we consider the situation when $\lambda^f = \lambda^b$, i.e., when the same channel is used twice. The lower bound on $r({\cal P}_2)$ is plotted in right panel in the same manner. }
\label{fig:corr_depol}
 \end{figure}

\subsubsection{Correlated Depolarising channel}

A single eavesdropper that is in close proximity to the receiver might be able to attack both channels connecting Alice with Bob1 and Bob2 simultaneously. In such a scenario, it is more feasible to consider a correlated noise model. In this section, we consider a fully correlated Pauli channel i.e a correlated depolarising channel given by
\begin{eqnarray}\label{eq:correlated_depol}
    {\cal E}^{f,b}_{\substack{A_1\rightarrow X_1 \\ A_2\rightarrow X_2 }} (\rho_{A_1A_2})= \left(1 - 3\frac{\lambda^{f,b}}{4}\right) \rho_{X_1X_2} +  \frac{\lambda^{f,b}}{4}\Big((\sigma_x \otimes \sigma_x) \rho_{X_1X_2} (\sigma_x \otimes \sigma_x) + (\sigma_y \otimes \sigma_y) \rho_{X_1X_2} (\sigma_y \otimes \sigma_y) \nonumber \\
    + (\sigma_z \otimes \sigma_z) \rho_{X_1X_2} (\sigma_z \otimes \sigma_z) \Big).
\end{eqnarray}

For the correlated noise model, we consider a more general scenario. This time the forward and backward channels are not the same. We choose parameter in Eq. (\ref{eq:correlated_depol}), as $\lambda^f$ for forward and $\lambda^b$ for backward transmission channel. 
The lower bound on the key rate is plotted in figure \ref{fig:corr_depol}. 
The left panel represents the lower bound on $r({\cal P}_1)$,  in the plane of $\lambda^f$ and $\lambda^b$, and the bound is symmetric. In comparison, there is an asymmetry for  $r({\cal P}_2)$, which is plotted in the right panel of figure \ref{fig:corr_depol}. 
 The symmetry of the plot implies that the lower bound on key rate for Alice to Bob1 is invariant under the swapping of forward and backward channels, which is not the case for the independent depolarising channel and the other local rate. Moreover, the the key rate for Alice to Bob2 decreases much faster for the effect of the noise in the backward channel compared to the noise in the forward channel, which is consistent with the result we have obtained for $r({\cal P}_2)$ for independent depolarising channel acting only in one side.

\subsection{Amplitude damping channel}
In the previous section, we have considered the depolarising noise model, which is a particular class of covariant noise \cite{CovariantChannel} which commutes with the encoding operations. Due to this commutation, we have observed that the conditional probability is independent of the auxiliary bit $s$. Furthemore, we find that its dependency on the conditioned event  $x,y;z,s$ is fixed.
 Hence, in this section, we will consider a completely different noise model, namely the amplitude damping channel, whose Kraus operators do not commute with the Pauli matrices. 
 
Suppose that the transmission channels are two independent amplitude damping channels that is

\begin{equation}\label{eq:amp_damp}
     {\cal E}^{f(b)}_{\substack{A_1\rightarrow X_1 \\ A_2\rightarrow X_2 }} = {\cal A}^{\gamma_1}_{A_1 \rightarrow X_1} \otimes {\cal A}^{\gamma_2}_{A_2 \rightarrow X_2},
\end{equation}
where $\gamma_1, \gamma_2 \in (0,1)$ are the noise parameters in the two channels, and the Kraus operator representation is 
\begin{equation}
    {\cal A}^{\gamma}(\rho) = \mathbb{A}_0\rho \mathbb{A}_0\dagger + \mathbb{A}_1\rho \mathbb{A}_1^\dagger
\end{equation}
with 
$$ \mathbb{A}_0=\left(
\begin{array}{cc}
 1 & 0 \\
 0 & \sqrt{1-\gamma } \\
\end{array}
\right), ~\text{and}~
\mathbb{A}_1=\left(
\begin{array}{cc}
 0 & \sqrt{\gamma } \\
 0 & 0 \\
\end{array}
\right).$$
We have numerically computed the lower bounds on both the local key rates and plotted them against the noise parameters $\gamma_1$ and $\gamma_2$, in figure \ref{fig:amp_damp}. We have assumed that the same channel has been used for the forward transmission and backward transmission of the 3-qubit quantum state.

Alike to the previous results, the key rate $r({\cal P}_1)$ is symmetric with respect to the noise parameters $\gamma_1$ and $\gamma_2$, whereas $r({\cal P}_2)$ turns out to be assymetric.


\begin{figure}[t]
\centering
 \includegraphics[width=1\columnwidth,keepaspectratio,angle=0]{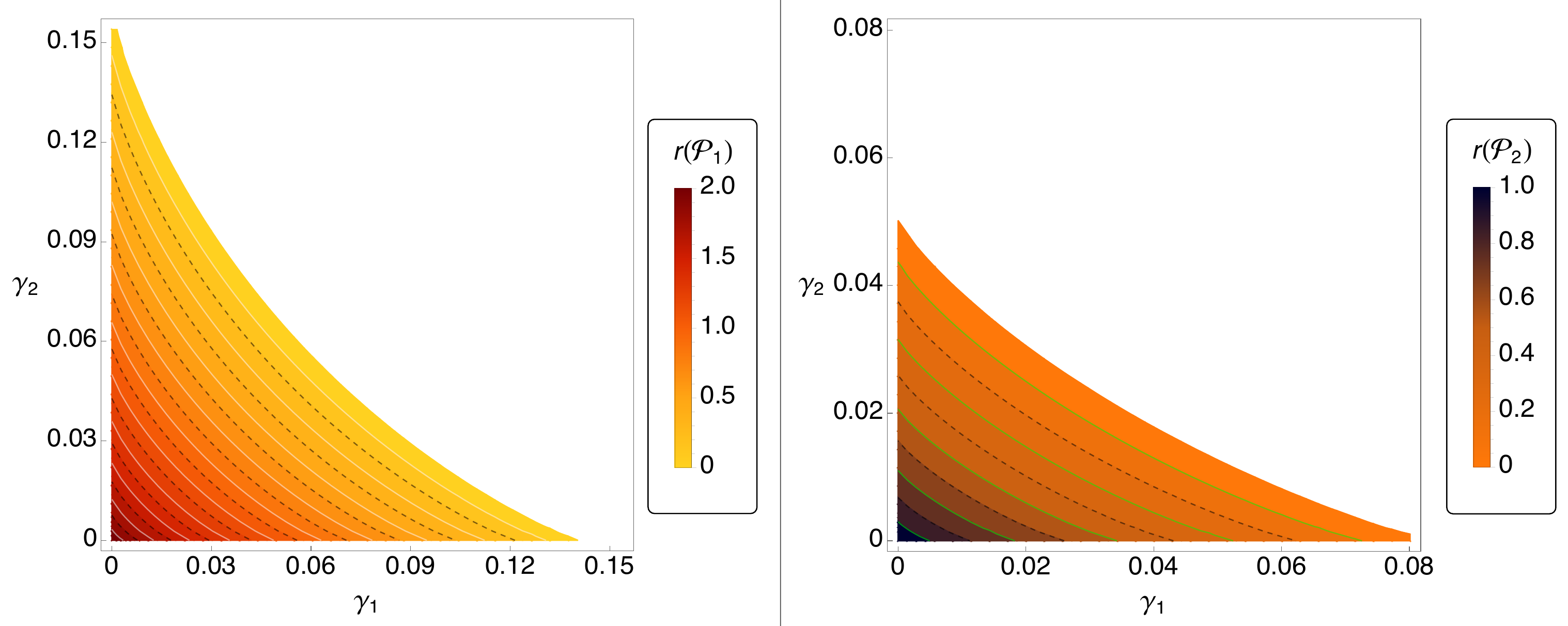}
\caption{Lower bounds on the local key rates for independent amplitude damping channels, when noise acts both in the forward and backward transmission channel. 
 $\gamma_1$ and $\gamma_2$ are the parameters of the independent depolarising channel as given in Eq. (\ref{eq:amp_damp}). }
\label{fig:amp_damp}
 \end{figure}

\section{Conclusions and Open Questions}\label{sec:discussion}
We have studied quantum key distribution over generalized multiple access channels between two senders and a single receiver. We have provided the formula for the achievable key rate region in this scenario. The lower bounds for secure key rate are given in terms of the entropic quantities of the random variables generated by a protocol which achieves the bounds. We have specified a protocol  based on the protocol of superdense coding. In that, we generalize the results of \cite{Beaudry2013} to the multipartite case.  The protocol consists of two parts. One is between the sender and a single receiver, and the second is between the sender and the other receiver. Based on distributed $GHZ$ state, it is natural in our scenario that one of the senders broadcasts part of his results. This, however, implies that 
we need to significantly modify the protocol of \cite{Beaudry2013} to the GMAC scenario.

One can expect similar
results for $n >2$ senders and a single
receiver. However, public communication required grows linearly with the number of senders in the proposed protocol. Thus, it is an important
open problem if this rate of communication is at all necessary. 

We note that other variants of $2$-$1$ SDC protocols can be provided
simply via a change of the measurements of both the sender and the receivers. One only needs to keep these measurements complementary with each other. This opens a path for designing a whole new class
of protocols with potentially higher key rates for a given
noisy channel. It would also be interesting to compare the classical information transmission with security constraints (aka SDC protocol) versus the classical capacity without security constraints. 


Recently a single-shot formula for dense coding has been
derived \cite{Korzekwa2019}. One can try to adapt it to the multipartite case together with the proof of security, which is also an important open problem.


We have discussed the fact that the symmetry of the GHZ state does not impose any restriction on which sender sends two bits of classical information and which one communicates one bit. However, the symmetry of the GHZ state breaks once the noise acts in the transmission channel. Therefore, the remaining open question is whether the honest parties can identify which channel is affected more by the noise during reconciliation. Connected to this is a question whether the public announcement of one bit from a selected sender reduces the effect of noise in the other party and if it is possible to disclose the bit for that particular sender whose channel is noisier.


%

\appendices

\section{Proof of the purification protocol}\label{sec:proof_puri}
In this section, we will prove that the unitary encoding performed by Bob1 and Bob2 on their shared part of the noisy GHZ state $\rho_{AB_1B_2}$, can be purified to a joint projective measurements ${\cal K}_{B_1B'_1} = \{K_{B_1B_1'}^{xy}\}_{x,y = 0}^1$ and ${\cal K}_{B_2B'_2} = \{K_{B_2B_2'}^{zs}\}_{z,s = 0}^1$ on $\rho_{AB_1B_2}$ and half of the  shared Bell state $\ket{\phi^+}_{B'_1X_1}$ and $\ket{\phi^+}_{B'_2X_2}$. Specifically 
\begin{eqnarray}\label{eq:encode_puri2}
\big(U_{X_1}^{x,y} \otimes U_{X_2}^{z,s} \rho_{AX_1X_2} U_{X_1}^{x,y\dagger} \otimes U_{X_2}^{z,s \dagger}\big) &=& 4^2 \times \text{tr}_{B_1B_1'}\text{tr}_{B_2B_2'}\big(\hat{K}^{x,y}_{B_1B'_1}\otimes \hat{K}^{z,s}_{B_2B'_2}(\rho_{AB_1B_2} \otimes |\phi^+\rangle\langle\phi^+|_{B'_1X_1} \nonumber \\
&& \hspace{5cm}\otimes |\phi^+\rangle\langle\phi^+|_{B'_2X_2}) \hat{K}^{{x,y}\dagger}_{B_1B'_1}\otimes  \hat{K}^{{z,s}\dagger}_{B_2B'_2}\big),
\end{eqnarray} 
where in the l.h.s. of Eq. (\ref{eq:encode_puri2}), we put the subscript $X_1$ and $X_2$, instead of $B_1$ and $B_2$ to make it consistent with the r.h.s.
Recall that the projectors $\{\hat{K}^{x,y}_{B_1B'_1}\}_{x,y = 0}^1$ and $\{\hat{K}^{z,s}_{B_2B'_2}\}_{z,s = 0}^1$ are given by
\begin{eqnarray}
    \hat{K}^{x,y}_{B_1B'_1} &=& \ket{B(x,y)}\bra{B(x,y)}_{B_1B'_1} = \frac{1}{2} \sum_{l,l' = 0}^1 (-1)^{(l\oplus l').y} \ket{l, x\oplus l}\bra{l', x\oplus l'}, \label{eq:PVxy} \\
    \hat{K}^{z,s}_{B_2B'_2} &=& \ket{B(z,s)}\bra{B(z,s)}_{B_2B'_2} = \frac{1}{2} \sum_{m,m' = 0}^1 (-1)^{(m\oplus m').s} \ket{m, z\oplus m}\bra{m', z\oplus m'}. \label{eq:PVzs}
\end{eqnarray}
Now express
\begin{eqnarray}
&&\rho_{AB_1B_2} = \sum_{i,j,k; i'j'k' = 0}^1 q(i,j,k;i'j'k')\ket{i,j,k}\bra{i'j'k'}_{AB_1B_2}, \label{eq:rho_expansion}\\
&& \ket{\phi^+}\bra{\phi^+}_{B_1'X_1} = \frac 12 \sum_{p;p' = 0}^1\ket{pp}\bra{p'p'}_{B_1'X_1},\\
&&\ket{\phi^+}\bra{\phi^+}_{B_2'X_2} = \frac 12 \sum_{q;q' = 0}^1\ket{qq}\bra{q'q'}_{B_2'X_2},
\end{eqnarray}
where $\tr(\rho_{AB_1B_2}) = \sum_{i,j,k} q(i,j,k;i,j,k) = 1 $. Hence the total system, 
\begin{eqnarray}
&&\rho_{AB_1B_2}\otimes |\phi^+\rangle\langle\phi^+|_{B'_1X_1} \otimes |\phi^+\rangle\langle\phi^+|_{B'_2X_2}  \\
&=& \frac 14 \sum_{i,j,k; i'j'k' = 0}^1 \sum_{p;p' = 0}^1 \sum_{q;q' = 0}^1 q(i,j,k;i'j'k')\ket{i,j,k}\bra{i'j'k'}_{AB_1B_2} \otimes \ket{pp}\bra{p'p'}_{B_1'X_1} \otimes \ket{qq}\bra{q'q'}_{B_2'X_2}, \nonumber \\
&=& \frac 14 \sum_{i,j,k; i'j'k' = 0}^1 \sum_{p;p' = 0}^1 \sum_{q;q' = 0}^1 q(i,j,k;i'j'k') \ket{ipq}\bra{i'p'q'}_{AX_1X_2} \otimes \ket{jp}\bra{j'p'}_{B_1B'_1} \otimes \ket{kq}\bra{k'q'}_{B_2B'_2}.
\end{eqnarray}
Now the r.h.s of Eq. (\ref{eq:encode_puri2}) can be written as 
\begin{eqnarray}
&& 4^2 \times \text{tr}_{B_1B_1'}\text{tr}_{B_2B_2'} \big(\hat{K}^{x,y}_{B_1B'_1}\otimes \hat{K}^{z,s}_{B_2B'_2}(\rho_{AB_1B_2} \otimes |\phi^+\rangle\langle\phi^+|_{B'_1X_1} \otimes |\phi^+\rangle\langle\phi^+|_{B'_2X_2}) \hat{K}^{{x,y}\dagger}_{B_1B'_1}\otimes  \hat{K}^{{z,s}\dagger}_{B_2B'_2}\big) \\
&=& 4^2 \times \frac 14 \sum_{\substack{ijk = 0\\ i'j'k' 0}}^1 \sum_{p;p' = 0}^1 \sum_{q;q' = 0}^1 q(i,j,k;i'j'k') \ket{ipq}\bra{i'p'q'}_{AX_1X_2} \nonumber  \\  
&& \hspace{2in} \bra{B(x,y)}\ket{jp}\bra{j'p'}\ket{B(x,y)} \bra{B(z,s)}\ket{kq}\bra{k'q'}\ket{B(z,s)} \\
&=& 4^2 \times \frac 14 \sum_{\substack{ijk = 0\\ i'j'k' 0}}^1 \sum_{p;p' = 0}^1 \sum_{q;q' = 0}^1 q(i,j,k;i'j'k') \ket{ipq}\bra{i'p'q'}_{AX_1X_2} \nonumber  \\  
&& \hspace{0.5in} \Big[\frac 12 \sum_{l,l' = 0}^1  (-1)^{(l \oplus l').y } \delta_{lj} \delta_{x\oplus l,p}\delta_{l'j'}\delta_{x\oplus l',p'}\Big] \Big[ \frac 12 \sum_{m,m' = 0}^1 (-1)^{ (m \oplus m').s} \delta_{mk} \delta_{z\oplus m,q}\delta_{m'k'}\delta_{z\oplus m',q'} \Big] \\
&=& \sum_{\substack{ijk = 0\\ i'j'k' 0}}^1 q(i,j,k;i'j'k') (-1)^{(j \oplus j').y ~\oplus ~(k \oplus k').s} \ket{i, j\oplus x,  k \oplus z}\bra{i', j'\oplus x, k'\oplus z}_{AX_1X_2}, \label{eq:final_form_puri}
\end{eqnarray}
where we use the extended expression of the projectors given in Eq. (\ref{eq:PVxy}) and (\ref{eq:PVzs}). 
To obtain the l.h.s of Eq. (\ref{eq:encode_puri2}), one should notice that both senders use the encoding according to the choice of Pauli matrices given in Eq. (\ref{Eq:encode_uni}) and their transform a quantum state $\ket{j}\bra{j'}$ as 
\begin{equation}\label{eq:uni_effect}
    U^{xy} \ket{j}\bra{j'} U^{xy \dagger} = (-1)^{(j \oplus j').y} \ket{j \oplus x}\bra{j'\oplus x}, ~\forall j,j',x,y \in \{0,1\}.
\end{equation}
Using Eq. (\ref{eq:uni_effect}) in Eq. (\ref{eq:final_form_puri}), we obtain
\begin{eqnarray}
 &&\sum_{\substack{ijk = 0\\ i'j'k' 0}}^1 q(i,j,k;i'j'k') \ket{i}\bra{i'}_A \otimes  U_{X_1}^{x,y} \ket{j}\bra{j'}_{X_1} U^{xy \dagger}  \otimes U_{X_2}^{z,s}  \ket{k}\bra{k'}_{X_2}  U_{X_2}^{z,s \dagger}, \nonumber  \\
 &=&  U_{X_1}^{x,y} \otimes U_{X_2}^{z,s } \rho_{AX_1X_2}U_{X_1}^{x,y\dagger} \otimes U_{X_2}^{z,s \dagger}
\end{eqnarray}
\hfill $\blacksquare$

\section{Proof of Eq. (\ref{Eq:probxyzs})}\label{sec:proof_equal_prob}
In this section we will prove that the probability of getting measurement outcome $\{x,y\}$ by Bob1 and $\{z,s\}$ by Bob2 satisfies Eq. (\ref{Eq:probxyzs}).
\begin{eqnarray}
&& p(xy;zs) = \sum_{ijk} p(ijk,xy;zs) = \sum_{ijk} \tr(\varrho_{EE'}^{ijk;xy;zs}) \nonumber \\
&=& \tr\left(\bra{B(x,y)}_{B_1B'_1} \otimes \bra{B(z,s)}_{B_2B'_2} \left( \ket{\Psi}\bra{\Psi}_{AA_1A_2B_1B'_1B_2B'_2EE'}\right)
    \ket{B(x,y)}_{B_1B'_1} \otimes \ket{B(z,s)}_{B_2B'_2} \right) \nonumber \\
 &=& \bra{B(x,y)}_{B_1B'_1} \otimes \bra{B(z,s)}_{B_2B'_2} \left(\rho_{AB_1B_2}   \otimes \tilde{\rho}_{B'_1A_1B'_2A_2}\right)
    \ket{B(x,y)}_{B_1B'_1} \otimes \ket{B(z,s)}_{B_2B'_2} \label{eq:prob-proof2} \\
    &=& \tr\left( \hat{K}^{x,y}_{B_1B'_1}\otimes \hat{K}^{z,s}_{B_2B'_2}(\rho_{AB_1B_2} \otimes \tilde{\rho}_{B'_1A_1B'_2A_2}) \hat{K}^{{x,y}\dagger}_{B_1B'_1}\otimes  \hat{K}^{{z,s}\dagger}_{B_2B'_2}\right) \nonumber \\
    &=& \tr\left( \hat{K}^{x,y}_{B_1B'_1}\otimes \hat{K}^{z,s}_{B_2B'_2}\left(\rho_{AB_1B_2} \otimes {\cal E}^b_{\substack{X_1 \rightarrow A_1\\ X_2\rightarrow A_2}} \left(|\phi^+\rangle\langle\phi^+|_{B'_1X_1} \otimes |\phi^+\rangle\langle\phi^+|_{B'_2X_2}\right) \right) \hat{K}^{{x,y}\dagger}_{B_1B'_1}\otimes  \hat{K}^{{z,s}\dagger}_{B_2B'_2}\right) \label{eq:prob-proof3} \\
    &=& \tr_{AA_1A_2} {\cal E}^b_{\substack{X_1 \rightarrow A_1\\ X_2\rightarrow A_2}} \left( \text{tr}_{B_1B_1'}\text{tr}_{B_2B_2'}\big( \hat{K}^{x,y}_{B_1B'_1}\otimes \hat{K}^{z,s}_{B_2B'_2}\left(\rho_{AB_1B_2} \otimes |\phi^+\rangle\langle\phi^+|_{B'_1X_1} \otimes |\phi^+\rangle\langle\phi^+|_{B'_2X_2}\right) \hat{K}^{{x,y}\dagger}_{B_1B'_1}\otimes  \hat{K}^{{z,s}\dagger}_{B_2B'_2}\big)\right) \nonumber \\
    &=& \frac{1}{16} \tr_{AA_1A_2} {\cal E}^b_{\substack{X_1 \rightarrow A_1\\ X_2\rightarrow A_2}} \left( U_{X_1}^{x,y} \otimes U_{X_2}^{z,s } \rho_{AX_1X_2}U_{X_1}^{x,y\dagger} \otimes U_{X_2}^{z,s \dagger} \right) = \frac{1}{16},
\end{eqnarray}
where in Eq. (\ref{eq:prob-proof2}), we use the fact that $\tr_{EE'}(\ket{\Psi}\bra{\Psi}_{AA_1A_2B_1B'_1B_2B'_2EE'}) = \left(\rho_{AB_1B_2}   \otimes \tilde{\rho}_{B'_1A_1B'_2A_2}\right)$, in  Eq. (\ref{eq:prob-proof3}) we used Eq. (\ref{eq:noiseonpuri}) and in the last equality we employed Eq. (\ref{eq:encode_puri2}).
 \hfill $\blacksquare$

\section{Calculation of the overlap between two measurement operators}\label{appen:overlap}
For any operator $A$, the infinity norm $||A||_\infty$ \cite{Berta2010,tomamichel2012framework}, is defined as 
\begin{equation}
    ||A||_\infty = \max_{\ket{\phi}, \braket{\phi|\phi} = 1} \bra{\phi} A \ket{\phi}.
\end{equation}

From Eq. (\ref{eq:overlap1}), we have 
\begin{eqnarray}
c &=& \max_{(i,j), (i'j')} ||\sqrt{(\hat{K}_s^{i',j'})_{AA_1A_2}}\sqrt{(\hat{T}^{i,j})_{AA_1A_2}}||_{\infty}^2 \nonumber \\
&=& \max_{(i,j), (i'j')} ||\sqrt{(\hat{T}^{i,j})_{AA_1A_2}}(\hat{K}_s^{i',j'})_{AA_1A_2}\sqrt{(\hat{T}^{i,j})_{AA_1A_2}}||_{\infty},    
\end{eqnarray}
where $(\hat{T}^{i,j})_{AA_1A_2} =  \ket{i}\bra{i}_A \otimes \ket{j_\vdash}\bra{j_\vdash}_{A_1} \otimes \mathbb{I}_{A_2}$
and 
$(\hat{K}_s^{i'j'})_{AA_1A_2} = \sum_{k'} (\hat{K}_s^{i'j'k'})_{AA_1A_2}$. Therefore 

\begin{eqnarray}
    && \sqrt{(\hat{T}^{i,j})_{AA_1A_2}}(\hat{K}_s^{i',j'})_{AA_1A_2}\sqrt{(\hat{T}^{i,j})_{AA_1A_2}} \\
    && = \ket{i}\bra{i}_A \otimes \ket{j_\vdash}\bra{j_\vdash}_{A_1} \otimes \mathbb{I}_{A_2} \left(\sum_{k'} \ket{G^s(i',j',k')} \bra{G^s(i',j',k')}_{AA_1A_2} \right) \ket{i}\bra{i}_A \otimes \ket{j_\vdash}\bra{j_\vdash}_{A_1} \otimes \mathbb{I}_{A_2}, \label{eq:overlap_AB1}
\end{eqnarray}
Now consider the partial inner product state
\begin{eqnarray}
    &&\bra{i}_A \otimes \bra{j_\vdash}_{A_1} \ket{G^s(i',j',k')}_{AA_1A_2} \\
    && = \frac 12 \bra{i}_A \otimes \sum_l (-1)^{j.l} \bra{l}_{A_1} \sum_{l'} (-1)^{l'.(j'\oplus s)} \ket{l', l' \oplus i', l' \oplus k'}_{AA_1A_2} \\
    && = \frac 12 (-1)^{i. (j \oplus j') \oplus i'.j \oplus i.s} \ket{i \oplus k'}_{A_2},
\end{eqnarray}
putting the value of this partial inner product in Eq. (\ref{eq:overlap_AB1}), we get 
\begin{eqnarray}
    && \sqrt{(\hat{T}^{i,j})_{AA_1A_2}}(\hat{K}_s^{i',j'})_{AA_1A_2}\sqrt{(\hat{T}^{i,j})_{AA_1A_2}}  \\
    &&= \frac 14 \ket{i}\bra{i}_A \otimes \ket{j_\vdash}\bra{j_\vdash}_{A_1} \otimes \sum_{k'} \ket{i \oplus k'}\bra{i \oplus k'}_{A_2} \\
    && = \frac 14 \ket{i}\bra{i}_A \otimes \ket{j_\vdash}\bra{j_\vdash}_{A_1} \otimes \mathbb{I}_{A_2}.
\end{eqnarray}
Hence, $c = \max_{(i,j), (i'j')} = \frac 14 || \ket{i}\bra{i}_A \otimes \ket{j_\vdash}\bra{j_\vdash}_{A_1} \otimes \mathbb{I}_{A_2}||_\infty = \frac 14$.

To calculate the overlap $ \tilde{c} $, given in  Eq. (\ref{eq:overlap2}), lets recall  
\begin{eqnarray}
     \tilde{c} 
     &=& \max_{k, k' } ||\sqrt{(\hat{\tilde{K}}_s^{k'})_{AA_1A_2}}\sqrt{(\hat{\tilde{T}}^k_s)_{AA_1A_2}}||_{\infty}^2 \nonumber \\
     &=& \max_{k,k'} ||\sqrt{(\tilde{T}^{k}_s)_{AA_1A_2}} \sum_{i'j'}(\hat{K}_s^{i'j'k'})_{AA_1A_2}\sqrt{(\tilde{T}^{k}_s)_{AA_1A_2}}||_{\infty},
 \end{eqnarray}
 where  $(\hat{K}_s^{i'j'k'})_{AA_1A_2}^{i,j,k} = \ket{G^s(i',j',k')} \bra{G^s(i',j',k')}_{AA_1A_2}$. Thus 
 \begin{eqnarray}
      && \sqrt{(\tilde{T}^{k}_s)_{AA_1A_2}} \sum_{i'j'}(\hat{K}_s^{i'j'k'})_{AA_1A_2}\sqrt{(\tilde{T}^{k}_s)_{AA_1A_2}} \\
      && =   \mathbb{I}_{AA_1}\otimes \ket{(k \oplus s)_\vdash}\bra{(k \oplus s)_\vdash}_{A_2} \left(\sum_{i'j'} \ket{G^s(i',j',k')} \bra{G^s(i',j',k')}_{AA_1A_2} \right)   \mathbb{I}_{AA_1}\otimes \ket{(k \oplus s)_\vdash}\bra{(k \oplus s)_\vdash}_{A_2}, \\
      && = {}_{A_2}\bra{(k \oplus s)_\vdash} \left( \ket{G^s(i',j',k')} \bra{G^s(i',j',k')}_{AA_1A_2} \right)  \ket{(k \oplus s)_\vdash}_{A_2} \otimes \ket{(k \oplus s)_\vdash}\bra{(k \oplus s)_\vdash}_{A_2}, \\
      && = \frac 12 \sum_{i'j'} \ket{B(i',j'+k)}\bra{B(i',j'+k)}_{AA_1} \otimes \ket{(k \oplus s)_\vdash}\bra{(k \oplus s)_\vdash}_{A_2}, \\
      && =  \frac 12 \mathbb{I}_{AA_1}\otimes \ket{(k \oplus s)_\vdash}\bra{(k \oplus s)_\vdash}_{A_2},
 \end{eqnarray}
 where we have used the fact that 
 \begin{eqnarray}
    {}_{A_2}\bra{(k \oplus s)_\vdash} G^s(i',j',k')\rangle_{AA_1A_2} &=& \frac 12 \sum_{m} (-1)^{(k \oplus s).m} {}_{A_2}\bra{m} \sum_l (-1)^{l.(j' \oplus s)} \ket{l, l \oplus i', l \oplus k'}_{AA_1A_2} \\
    &=& \frac 12 \sum_l (-1)^{l.(j' \oplus k) \oplus k.k'} \ket{l, l\oplus i}_{AA_1} = \frac{1}{\sqrt{2}}(-1)^{k.k'} \ket{B(i, j \oplus k')}
 \end{eqnarray}
 Finally, 
 \begin{eqnarray}
    \tilde{c} 
     = \frac 12 \max_{k', k} || \mathbb{I}_{AA_1}\otimes \ket{(k \oplus s)_\vdash}\bra{(k \oplus s)_\vdash}_{A_2} ||_\infty = \frac 12.
\end{eqnarray}

\section*{Acknowledgment}

TD thanks Shilpa Samaddar for the beautiful illustrations of Alice, Bob and Eve. 
TD and KH acknowledge grant Sonata Bis 5 (grant number: 2015/18/E/ST2/00327) from the National Science Center. 
TD and KH acknowledge partial support by Foundation for Polish Science (FNP), IRAP project ICTQT, contract no. 2018/MAB/5, cofinanced by EU Smart Growth Operational Programme. RP acknowledges support from EPSRC (UK).

\ifCLASSOPTIONcaptionsoff
  \newpage
\fi

\bibliographystyle{IEEEtran}
\bibliography{reference}

%








\end{document}